\documentclass{article}
\usepackage{amsmath,amsthm,amssymb,array,authblk, booktabs,ccaption,color,dcolumn,float,graphicx,harvard,mathdots,mathtools,microtype,relsize,ulem,diagbox}
\usepackage{clipboard}
\newclipboard{output-myclipboard}

\usepackage[table]{xcolor}
\usepackage[UKenglish]{isodate}
\usepackage[colorlinks,breaklinks]{hyperref}
\usepackage[doublespacing]{setspace}
\usepackage[margin=1in]{geometry}
\newtheorem{assumption}{Assumption}
\newtheorem{theorem}{Theorem}
\newtheorem{lemma}{Lemma}
\newtheorem{remark}{Remark}

\begin{document}

\title{Testing for equal predictive accuracy with strong dependence
}
\author[1]{Laura Coroneo}
\author[2,1]{Fabrizio Iacone\thanks{Corresponding author. Contact: Department of Economics, Management and Quantitative Methods, Universit\`{a} degli Studi di Milano Via Conservatorio 7, 20122 Milano. Email: fabrizio.iacone@unimi.it. The authors thank the Editor, Esther Ruiz, the associate editor and two anonymous referees for their constructive and insightful comments, which helped improve the paper substantially. We also thank Roberto Casarin, Gergley G\'{a}nics, Raffaella Giacomini, Liudas Giraitis, Anders Kock, Andrey Vasnev, and participants to the 7th RCEA Time Series Workshop (University of Milan Bicocca, 2021), to the 42nd International Symposium on Forecasting (ISF 2022 - University of Oxford), and the Bergamo Workshop of Econometrics and Statistics for useful comments.}}
\affil[1]{University of York}
\affil[2]{Universit\`{a} degli Studi di Milano}
\date{\today}

\maketitle
\begin{abstract}
\noindent 
We analyse the properties of the \citeasnoun{diebold1995comparing} test in the presence of autocorrelation in the loss differential.  We show that the power of the \citeasnoun{diebold1995comparing} test decreases as the dependence increases, making it more difficult to obtain statistically significant evidence of superior predictive ability against less accurate benchmarks. We also find that, after a certain threshold, the test has no power and the correct null hypothesis is spuriously rejected. Taken together, these results caution to seriously consider the dependence properties of the loss differential before the application of the \citeasnoun{diebold1995comparing} test.

\end{abstract}
\par \indent \textit{Total Words}: 5,744
\clearpage
\section{Introduction}\label{sec_intro}
Accurate forecasts are extremely important for forward-looking decision making. Weather forecasts often have a dedicated section even on the daily news, and  predictions of the diffusion of the COVID-19 pandemic have had a critical impact on the life of most people. In economics, decisions over individual savings, firm-level investments, government fiscal policies and central bank monetary policies rely on forecasts of, among others, future economic activity and price levels. 

To discriminate  between good and bad forecasts, \citeasnoun{diebold1995comparing} [DM hereafter] suggested  comparing alternative forecasts using a test for equal predictive ability. The DM test  is based on a loss function associated with the forecast errors of each forecast, and allows to  test the null hypothesis of zero expected loss differential between two competing forecasts.  This approach takes forecast errors as model-free, and the test is valid also when the forecasts are produced from unknown models, as for example from forecast survey data. In addition, if the objective is to compare forecasting methods as opposed to forecasting models, then \citeasnoun{giacomini2006tests} showed that in an environment with asymptotically non-vanishing estimation uncertainty the DM test can still be applied. 

The DM test allows us to test for equal predictive accuracy using any loss function, and the test statistic is asymptotically valid for contemporaneously correlated, serially correlated, and non-normal forecast errors. The test relies on the assumption that the loss differential is weakly dependent.
The rationale for this assumption is that, under mean squared error (MSE) loss, optimal $q$-step ahead forecasts should generate at most MA($q-1$) errors. Thus if the considered forecast approximates the optimal forecast, its forecast errors should not be  too correlated over time, although  dependence beyond the MA($q-1$) boundary may take place. In practice, forecasts with errors that are fairly correlated can  occur  not only when  the considered forecast fails to approximate the optimal forecast under MSE loss,  but also when 
the forecast is optimal under an alternative loss function \citeaffixed{patton2007properties}{see} or it is evaluated on a relatively short sample. 

Still, we can encounter situations in which the loss differential is highly autocorrelated even in the presence of a  prediction with weakly dependent forecast errors and in samples of moderate size. This can happen when the DM test is used to compare the predictive ability of a selected forecast against a naive benchmark. 
This is a common practice, as naive benchmarks are cost-effective and readily available at any time, so they provide a standard reference for comparisons. Using simple benchmarks  allows us to understand the added value of a specific forecasting technique, as it is desirable that predictions from sophisticated forecasting methods (for example complex models or expensive surveys) are more accurate than naive benchmarks.  However, naive forecasts  may in some cases generate relevant autocorrelation  in the loss differential.    

In this paper, we study the performance of the DM test when the assumption of weak autocorrelation of the loss differential does not  hold. We characterise strong dependence as local to unity as in \citeasnoun{Phillips1987LocUnity} and \citeasnoun{PhillipsMagd2007MildlyIntegrated}. This definition is at odds with the more popular characterisation in the literature that treats strong autocorrelation and long memory as synonyms. Local to unity, however, seems well suited to derive reliable guidance when the sample is not very large, as it is the case in many applications in economics. With this definition the strength of the dependence is determined also by the sample size: a stationary AR(1) process with root close to unity may be treated as weakly dependent in a very large sample, but standard asymptotics may be a poor guidance for cases with smaller samples, and local to unity asymptotics may be more informative. We show that the power of the DM test decreases as the dependence increases, making it more difficult to obtain statistically significant evidence of superior predictive ability against less accurate benchmarks.  We also find that after a certain threshold the test has no power and the correct null hypothesis is spuriously rejected.   \Copy{R2c7t}{These results caution us to seriously consider the dependence properties of the loss differential before the application of the DM test, especially when naive benchmarks are considered. In this respect, a unit root test could be a valuable diagnostic for the preliminary detection of critical situations.}\label{R2c7}

To  illustrate  the  problems  associated  with  the  DM  test  when  there  is  dependence in the loss differential,  we consider the case in which an AR(1)  forecast  for inflation in the Euro Area  is compared to two naive benchmarks: a constant 2\% prediction (that represents the inflation target in the Euro Area) and a rolling average prediction. These benchmark predictions have highly dependent forecast errors. As a consequence, the loss differential is dependent and the DM test fails to reject the null of equal predictive accuracy, even if the benchmarks are less accurate than the  AR(1) forecast  for short forecasting horizons.

In the literature, there has been some attention to the issue of forecast evaluation in presence of persistence.  \citeasnoun{corradi2001predictive} examined the DM statistic in the presence of cointegration, whereas  \citeasnoun{rossi2005testing} examined the effect of high persistence on the loss differential. \citeasnoun{mccracken2020diverging} provided an example  to show that using a fixed and finite estimation window  can result in loss differentials that depend on the first observations, so that the time series of the loss differential in the DM test is not ergodic for the mean. These works considered a framework with  parameter estimation error;  instead \citeasnoun{clark1999finite}, 
\citeasnoun{Khalaf2017} and
\citeasnoun{kruse2019comparing} took forecast errors as primitives. \citeasnoun{clark1999finite} and 
\citeasnoun{Khalaf2017}, in particular, provided convincing simulation evidence that the DM test is incorrectly sized in the presence of dependence. \citeasnoun{giacomini2006tests} and 
\citeasnoun{coroneo2020comparing} showed that under benign forms of weak dependence, the correct size may be restored using bootstrap or fixed smoothing asymptotics, respectively, but results in \citeasnoun{Khalaf2017} suggest that even bootstrap is only a useful and reliable guidance when the depedence is not too strong, relative to the sample size.
On the other hand \citeasnoun{kruse2019comparing}
derived the properties of the DM test in the presence of long memory using standard asymptotics, and memory and autocorrelation consistent standardisation. \citeasnoun{kruse2019comparing} had a sample of 4883 observations, so long memory seems a reasonable modelling strategy in their case. However, this approach cannot be applied to moderate sample size, such as the ones typically encountered in macroeconomic forecasting.

Of course, the DM test can be seen as a particular application of the standard $t$-test on the mean in presence of dependence. A certain level of persistence can be accommodated for the $t$-test using bootstrap or alternative asymptotics, see for example \citeasnoun{gonccalves2011block} for the former, and
\citeasnoun{kiefer2005new}, \citeasnoun{sun2014let}, \citeasnoun{sun2014fixed}, \citeasnoun{lazarus2018har}  for the latter. Following \citeasnoun{muller2014hac} or \citeasnoun{GiraitisPhillips2012Local},
it is clear that the even the limit distribution in \citeasnoun{kiefer2005new} and related works does not provide a useful guidance when the persistence is too strong in relation to the sample size. \citeasnoun{muller2014hac}  does provide an algorithm to perform a reliable test in that case, although the sample size required is larger than the samples that are usually available in forecast evaluation exercises. \label{Rass2} As a promising alternative option, \citeasnoun{henzi2021valid}  and \citeasnoun{choe2023comparing} recently proposed procedures based on sequential testing that do not require assumptions on the dependence of the process. These new procedures may therefore be robust even in situations of dependence, however the assumption of bounded loss differentials may limit their applicability.

As it is clear from this review of the literature, the core of the statistical results discussed here should not come as a surprise. However, their implication  in the context of testing for equal predictive ability remains relevant, in particular  when good forecasts are compared against poor benchmarks, with relevant autocorrelations, as  our results suggest that in these cases the blind application of the DM test leads to incorrect conclusions. Even more importantly, it may be more difficult to reject the null hypothesis when good forecasts are compared to poor competitors, than when the same good forecasts are compared to competitors that are nearly as good. This perverse and undesirable feature of the DM test should be kept in mind in forecast evaluation exercises.

The paper is organised as follows. We formally introduce the DM test in Section \ref{sec_DMtest}, and derive the limit properties of the DM statistics in presence of dependence in Section \ref{sec_theory}. We investigate the practical implication of our theoretical findings in a Monte Carlo exercise (Section \ref{sec:MC}) and in the empirical application (Section \ref{sec:emp}). Details on the assumptions of the DGP and formal derivations are in the Appendix.

\section{DM test}\label{sec_DMtest}
The DM test was introduced to compare two forecasts of a time series, according to a user chosen loss metric. For $t=\{1, \hdots, T\}$, denoting the forecast errors as $e_{1,t}$ and $e_{2,t}$, respectively, and the loss function $L(.)$, \citeasnoun{diebold1995comparing} 
consider the loss differential  
\begin{equation} \label{LossDiff}
    d_t=L(e_{1,t})-L(e_{2,t})
\end{equation}
and test the null hypothesis of equal predictive ability $H_0: E(d_t)=0$ against the alternative $H_1: E(d_t) \neq 0$.
The key assumptions by \citeasnoun{diebold1995comparing} and \citeasnoun{diebold2015comparing} are that 
$d_t$ is stationary and weakly dependent, and that the average loss, $\overline{d}=\frac{1}{T} \sum_{t=1}^{T}{d_t}$, follows a Central Limit Theorem. In particular, denoting $\mu=E(d_t)$, it is assumed that $\sqrt{T}{(\overline{d}-\mu)} \rightarrow_d {N(0,{\sigma^2})}$ as $T \rightarrow \infty$, where $0< \sigma^2<\infty $
is the long run variance of $d_t$. 

Thus, inference on $E(d_t)$ can be based on the normalised limit \begin{equation} \label{standardCLT}
    \sqrt{T}\frac{(\overline{d}-\mu)}{\sigma} \rightarrow_d {N(0,1)}.
\end{equation}
Denoting $\widehat{\sigma}^2$ an estimate of $\sigma^2$,  \hspace{0.1cm}
$\widehat{\sigma}=\sqrt{\widehat{\sigma}^2}$, then the classical DM test uses the statistic
\begin{equation*}
DM=\sqrt{T}\frac{\overline{d}}{\widehat{\sigma}} 
\end{equation*}
where the null hypothesis of equal predictive ability is rejected at 5\% significance level against a two-sided alternative if the realization of $|DM|$ is above the 1.96 threshold. 

The original DM test exploits the consistency of $\widehat{\sigma}^2$ to justify the standard normal as the limit distribution under the null. This may generate rather poor size performance  in finite sample, see \citeasnoun{diebold1995comparing} and also \citeasnoun{clark1999finite}. 
With fixed smoothing asymptotics, the limit for $\widehat{\sigma}^2$ is derived under alternative asymptotics, accounting for the distribution of $\widehat{\sigma}^2$ and, as a consequence,  the DM statistic does not have limit standard normal distribution, but the alternative limit provides a better approximation of the distribution of the DM statistic in finite samples.
As the alternative distribution depends on the way $\sigma^2$ is estimated, we consider two cases: the weighted autocovariance estimate using the Bartlett kernel and the weighted periodogram estimate using the  Daniell kernel. 

Denoting by $\widehat{\gamma}_l$ the sample autocovariance of lag $l$,  the weighted autocovariance estimate of the long run variance using the Bartlett kernel is
\begin{equation}\label{bartlett}
\widehat{\sigma}_A^2=\widehat{\gamma}_0+2 \sum_{l=1}^M {\frac{M-l}{M}\widehat{\gamma}_l}.
\end{equation} 
where $M$ is a user-chosen bandwidth parameter, and
\begin{equation*}
DM_A=\sqrt{T}\frac{\overline{d}}{\widehat{\sigma}_A}.
\end{equation*}
Under $H_0$, 
\begin{align*}
  & DM_A\rightarrow_d N(0,1), \text{ if }
  1/M+M/T  \rightarrow 0 \text{ as } T\rightarrow \infty,\\
 & DM_A\rightarrow_d \Phi_A(b),  \text{ if }
  M/T  \rightarrow b \in (0,1] \text{ as } T\rightarrow \infty
\end{align*}
where the distribution of $\Phi_A(b)$ depends on $b$; this is characterised in 
\citeasnoun{kiefer2005new}, where relevant quantiles are also provided.  

For the weighted periodogram estimate, 
denoting $w(\lambda)=\frac{1}{\sqrt{2 \pi T}} \sum_{t=1} ^ {T} d_t e^{i \lambda t}$ the Fourier transform of $d_t$ at frequency $\lambda$, and $I(\lambda)=|w(\lambda)|^2$ as the periodogram, the Daniell weighted periodogram estimate is
\begin{equation} \label{Daniell}
    \widehat{\sigma}^2_P=\frac{2 \pi}{m}\sum_{j=1} ^m I(\lambda_j)   
\end{equation}
where, for integer $j$, $\lambda_j=\frac{2 \pi j}{T}$ are the Fourier frequencies and $m$ is a user-chosen bandwidth parameter. The test statistic is then given by
\begin{equation}\label{fix_m}
DM_P = \sqrt{T}\frac{\overline{d}}{\widehat{\sigma}_P}.
\end{equation}
Under $H_0$,
\begin{align*}
  & DM_P\rightarrow_d N(0,1), \text{ if }
  1/m+m/T  \rightarrow 0 \text{ as } T\rightarrow \infty,\\
 & DM_P\rightarrow_d t_{2m},  \text{ if }
  m \text{ is fixed }  \text{ as } T\rightarrow \infty
\end{align*}
where $t_{2m}$ is the Student's $t$-distribution with $2m$ degrees of freedom, \label{Rass} see \citeasnoun{coroneo2020comparing} for more details.

\section{The DM statistic  with dependence}\label{sec_theory}
The key assumption in the construction of the DM test is  that  $d_t$ is  weakly dependent. This assumption seems reasonable in the context of forecasts, as it is well known that, under MSE loss, optimal $q-$step ahead forecasts should be at most $MA(q-1)$. For this very reason,  \citeasnoun{diebold1995comparing} considered estimating $\sigma^2$ using only the first $q-1$ autocovariances of $d_t$, and verified that this assumption was met in the data in the empirical application that they presented. 

However, in practice, it is not uncommon to have strong autocorrelation in the loss differential.
This can happen in presence of forecasts that are optimal under alternative loss functions, or when the forecast evaluation sample $T$ is short. In addition, it is  common practice to apply the DM test to test for equal predictive accuracy of a selected forecast against a naive benchmark, resulting in dependent forecast errors and loss differential, possibly even strongly  autocorrelated. Section~\ref{sec:MC} contains an example illustrating how stochastically trending behavior may appear in the loss differentials.
\label{p4.2}

Denoting $\mu=E(d_t)$ and $y_t=d_t-\mu$, so that
\begin{equation}
    d_t= \mu + y_t \label{dt},
\end{equation}
we assume that 
\begin{equation}
           y_t=\rho_T y_{t-1}+u_t \label{ar1}
\end{equation}
where $u_t$ is a zero mean, weakly dependent process with long run variance $\omega$. We consider two different models for $\rho_T$: in subsection \ref{Subsec_Local} we discuss the local-to-unity AR(1) approximation as in \citeasnoun{Phillips1987LocUnity} (alongside with the standard unit root model); in subsection \ref{Subsec_mildly} we consider the moderate deviations from a unit root as in \citeasnoun{PhillipsMagd2007MildlyIntegrated}. These models are a convenient representation of  dependence for $y_t$ when the dimension in $T$  is relatively short, as it is indeed the case in many empirical studies. In both cases, we refer to Appendix~\ref{app_proofs} for a detailed presentation and discussion of the assumptions, and for the derivation of the results.

\subsection{Local to unity autocorrelation} \label{Subsec_Local}
In this case, we assume for $\rho_t$ in \eqref{ar1}
 \begin{equation} \label{rho_T}
    \rho_T=e^{c/T}  \quad \text{with  } c \leq 0.   
\end{equation}
When $c$ is in the neighbourhood of $0$, $\rho_T$ is approximated as $1+{c/T}$, i.e. $    \rho_T \sim 1+{c/T}$  as $ c \rightarrow 0.$
When $c=0$ then the process $y_t$ has a unit root, and it is initialised setting the initial condition $y_0=O_p(1)$.

Our model is completed by formalising the assumptions on $u_t$.
\begin{assumption}\label{ass1}
Let $\varepsilon_t$ be independent and identically distributed (iid) random variables, with $E(\varepsilon_t)=0$,
$E(\varepsilon_t^2)=\varsigma^2$.
Then, assume that $u_t=\sum_{j=0}^{\infty}
\psi_j \varepsilon_{t-j}   
$ is such that
\begin{equation*}
\left(\sum_{j=0}^{\infty}
\psi_j \right)^2>0
\text{,  }\sum_{j=0}^{\infty}
 j^{1/2}|\psi_j|<\infty.
\end{equation*}
\end{assumption}

Denoting $g(\lambda)$ as the spectral density of $u_t$,  Assumption~\ref{ass1} implies that $g(0)>0$ and that $g(\lambda)<\infty$ uniformly in $\lambda$. Assumption~\ref{ass1} is sufficient to establish the functional central limit theorem (FCLT) for a stationary, weakly dependent linear process as in  \citeasnoun{PhillipsSolo1992}, see remark 3.5 of \citeasnoun{PhillipsSolo1992}, and notice that condition $\sum_{j=0}^{\infty}
 j^{1/2}|\psi_j|<\infty$ implies (16) of \citeasnoun{PhillipsSolo1992}, as discussed on page 973. 

Define 
\begin{equation*}
    J_c(r)= \int_0^r e^{(r-s)c} dW(s)
\end{equation*}
where $W(r)$ is a standard Brownian motion. The process $J_c(r)$ is a Ornstein-Uhlenbeck process: for given $r$ it is normally distributed (when $c=0$ the Ornstein-Uhlenbeck process is the standard Brownian motion). We refer to \citeasnoun{Phillips1987LocUnity} for a detailed discussion, but we state some important results from Lemma 1 of \citeasnoun{Phillips1987LocUnity}:
\begin{align}
  & T^{-1/2} y_{\lfloor rT \rfloor} \rightarrow 
  \omega J_c(r) \label{Lemma1aP}\\ 
  & T^{-1/2}  \overline{y} \rightarrow \omega \int_0 ^1 J_c(r)dr \label{Lemma1bP}\\
  & T^{-2} \sum_{t=1}^T {y_t^2} \rightarrow \omega^2 \int_0 ^1 J_c(r)^2dr \label{Lemma1cP}
\end{align}
where $J_c(r)=W(r)$ when $c=0$. The limit in (\ref{Lemma1aP}) follows proceeding as in \citeasnoun{Phillips1987LocUnity} but using the FCLT for linear processes instead of for mixing sequencies (also see \citeasnoun{Chan_and_Wei});  (\ref{Lemma1bP}) and (\ref{Lemma1cP}) are then due to  the continuous mapping theorem. The result for $\overline{y}$ in particular means that the sample average is not consistent in the neighbourhood of a unit root. 

Denoting 
\begin{equation*}
    \overline{J_c}=\int_0 ^1 J_c(r)dr,
\end{equation*} 
we can now establish the limit properties of the $DM_A$ and $DM_P$ statistics.

\begin{theorem} \label{DF-standard}
For $d_t$ generated as in \eqref{dt}-\eqref{rho_T}, and under Assumption~\ref{ass1}, 
\begin{itemize}
    \item[Case A:] For $1/M+M/T \rightarrow 0$ as $T \rightarrow \infty$
\begin{equation}\label{SpuriousDM_A}
\frac{\sqrt{M}}{\sqrt{T}} DM_A \rightarrow_d \frac{\overline{J}_c}{\sqrt{\int_0^1 \left( {J}_c(r)^2-\overline{J}_c^2 \right) dr}}
\end{equation}
    \item[Case P:] For $1/m+m/T \rightarrow 0$ as $T \rightarrow \infty$
    \begin{equation}\label{SpuriousDM_P}
    \frac{1}{\sqrt{m}} DM_P \rightarrow_d  \frac{{\overline{J}_c}}{\sqrt{\frac{1}{2}\int_0 ^1  \left(J_c(r)^2-\overline{J_c}^2\right) dr}} 
\end{equation}
\end{itemize}
\end{theorem}

We refer to Appendix~\ref{app_proofs} for a more detailed derivation of this and other results in this section. In view of \eqref{SpuriousDM_A} or \eqref{SpuriousDM_P}, 
as $M/T \rightarrow 0$ or $m \rightarrow \infty$ the DM test statistic diverges even when the null hypothesis is correct, thus giving spurious evidence of superior predictive ability. As we interpret the local to unit root as an approximation of an AR(1) in finite sample, this result suggests that the DM test diverges in presence of a root that is stationary but close to 1.  

Next, we present the limit of the $DM_A$ and $DM_P$ statistics using fixed smoothing asymptotic. \\
Denoting
\begin{align*}
& \widetilde{J}_c(r)={J}_c(r)-r {J}_c(1)\\
& Q_A(b)=\frac{2}{b} \int_0^1 \widetilde{J}_c(r)^2dr-\frac{2}{b}\int_0^{1-b} {\widetilde{J}_c(r)\widetilde{J}_c(r+b)}dr\\
& Q_P(j)= (2 \pi j)^2 \left\{ \left( \int_{0}^{1} sin(2 \pi j r) \widetilde{J}_c(r) dr \right)^2
+ \left( \int_{0}^{1} cos(2 \pi j r) \widetilde{J}_c(r) dr \right)^2   \right\}
\end{align*}

\begin{theorem} \label{DF-fixed}
\textit{For $d_t$ generated as in \eqref{dt}-\eqref{rho_T}, and under Assumption~\ref{ass1},
\begin{itemize}
\item[Case A:] For $M/T \rightarrow b\in(0,1]$ as $T \rightarrow \infty$,
\begin{equation}\label{inconsistentDMA}
DM_A \rightarrow_d \frac{\overline{J}_c}{\sqrt{Q_A(b)}}
\end{equation}
    \item[Case P:] For $m$ fixed as $T \rightarrow \infty$ then 
\begin{equation}\label{inconsistentDMP}
   DM_P \rightarrow_d \frac{\overline{J_c}}{\sqrt{\frac{1}{m} \sum_{j=1}^m Q_{P}(j)}}.
\end{equation}
\end{itemize}}
\end{theorem}

When $c=0$, so that the process $y_t$ is characterised by a unit root, the limit distribution $\widetilde{J}_c(r)$ is replaced by
$\widetilde{W}(r)={W}(r)-r {W}(1)$
where $W(r)$ is a standard Brownian motion, and $\overline{J_c}$ is replaced by $\overline{W}=\int_0 ^1 W(r)dr$.
 
The limits in \eqref{inconsistentDMA} and \eqref{inconsistentDMP} exhibit the self-normalization property as in \citeasnoun{Shao2015Selfnormalization}. It is interesting to compare the limits in Theorem \ref{DF-standard} and in Theorem \ref{DF-fixed} with results in the literature. It is well known that the standardised mean is diverging in presence of strongly autocorrelated series when $m \rightarrow \infty$ is assumed (and an analogue result would hold when $M/T \rightarrow 0$), but not when $m$ is assumed fixed, see for example \citeasnoun{McElroyPolitis2012ET} and \citeasnoun{hualde2017fixed}. In the context of local to unity process, this result was established, for example, in  \citeasnoun{Sun2014EssayPCB}. The advantages of series variance estimators of the long run variance in presence of autocorrelation is also discussed in \citeasnoun{MULLER2007LRVestimates}. Our interest in the limits in Theorem \ref{DF-standard} and in Theorem \ref{DF-fixed} is, however, slightly different, as we do not see these as alternative limits under different asymptotics, but rather as guidance that show properties of the DM test for relatively small and large $m$ (large $M/T$ and small $M/T$).  The limits in Theorem \ref{DF-standard} and in Theorem \ref{DF-fixed} are only relevant in the sense that we can establish whether the test statistic is divergent or not.

\begin{remark}\label{REM_Local}
Results in  in Theorem \ref{DF-standard} and in Theorem \ref{DF-fixed} indicate  that: \begin{enumerate}
\item With relatively large values for $m$ (or equivalently small $M/T$), the DM statistics diverges even under $H_0$ (spurious significance).
\item With small values for $m$ (or equivalently large $M/T$),
the DM statistics does not diverge even under $H_1$, so the test is not consistent. However, as the distribution in \eqref{inconsistentDMP} has much thicker tails than a $t_{2m}$ distribution (and similarly for the distribution in \eqref{inconsistentDMA} with respect to the $\Phi_A(b)$ distribution),  then it is still possible (and indeed it may be frequent) to have many spurious rejections of the null hypothesis. 

\item The limits in Theorem \ref{DF-standard} and in Theorem \ref{DF-fixed} hold regardless of whether $\mu=0$ or $\mu \neq 0$, so they are not affected by whether the null hypothesis is correct or not. This means that  the DM test cannot discriminate between null and alternative hypotheses in this case. 
\end{enumerate}
\end{remark}

\subsection{Moderate deviations from unit root} \label{Subsec_mildly}
As $c$ in \eqref{rho_T} varies between $-\infty$ and $0$, it is possible to use the theory from subsection \ref{Subsec_Local} for any AR(1) model with positive autocorrelation. However, the limits in Theorem \ref{DF-standard} and in Theorem \ref{DF-fixed} may not provide a valuable guideline when $\rho_T$ is not in fact in the very close neighbourhood of unity. For this situation,  \citeasnoun{PhillipsMagd2007MildlyIntegrated} generalise $\rho_T$ to moderate deviations from the unit root. We simplify the model slightly and consider   
\begin{equation} \label{rho_mild}
    \rho_T=1+c/T^\alpha \text{  for  }\alpha \in (0,1), \text{ and } c<0.
\end{equation}
 Moderate deviations from the unit root following \eqref{rho_mild} are also discussed in \citeasnoun{PM2007Chapterinbook}.  \citeasnoun{GiraitisPhillips2012Local} provide a generalisation of some results under a weaker condition, similar to $(1-\rho_T)T\rightarrow \infty$.
We strengthen Assumption~\ref{ass1} slightly, as 
 \begin{assumption}\label{ass2} Assume that
\begin{equation*}
\sum_{s=j}^{\infty}
 |\psi_s|<C j^{-1-a} \text{  for } j \geq 1
\end{equation*}
for $a>2$.
\end{assumption}

Under~\eqref{rho_mild}, $\overline{d}$ is a consistent estimate of $\mu$ only when $\alpha \in (0,1/2)$, but the CLT in \eqref{standardCLT} still holds for any $\alpha \in (0,1)$, see  Theorem 2.1 and the discussion on page 168 of \citeasnoun{GiraitisPhillips2012Local}. Recalling that, for any $T$, $\sigma^2=(1-\rho_T)^{-2} \omega^2$ and noticing that this is proportional to $T^{2 \alpha}$ in large samples, the rate of convergence of the CLT is  reduced to $\sqrt{T^{1-2\alpha}}$ (the theory does not cover the $\alpha=1$ case, but notice that $\sqrt{T^{1-2\alpha}} \rightarrow T^{-1/2}$  as $\alpha \rightarrow 1$ and this is the rate in \eqref{Lemma1bP}, suggesting a proximity of the two representations; the extension of \eqref{Lemma1aP} under \eqref{rho_mild} is explored more in \citeasnoun{PM2007Chapterinbook}). 

\begin{theorem} \label{DMforMildly}
For $d_t$ generated as in \eqref{dt}-\eqref{ar1} and~\eqref{rho_mild},  under Assumption~\ref{ass2}, , for any $\alpha \in (0,1)$:

\begin{itemize}

\item[Case A:] For $1/M+M/T^\alpha \rightarrow 0$ as $T \rightarrow \infty$:
\begin{equation}\label{MildLimit2A}
(-c)^{1/2} (M/T^\alpha)^{1/2} \frac{\overline{d}-\mu}
   {\widehat{\sigma}_A} \rightarrow_d N(0,1)
\end{equation}

\item[Case P:] For $1/m+m/T \rightarrow 0$ as $T \rightarrow \infty$, 
\begin{align}
   &\text{if } mT^{\alpha-1} \rightarrow 0 : \ \sqrt{T}\frac{\overline{d}-\mu}
   {\widehat{\sigma}_P} \rightarrow_d {N(0,1)} \label{MildLimit1}  \\
   &\text{if } m T^{\alpha-1} \rightarrow \infty : \  {(mT^{\alpha-1})^{-1/2}} 
   {1/2 (-c)^{1/2}} \left\{ \sqrt{T}\frac{\overline{d}-\mu}
   {\widehat{\sigma}_P} \right\} \rightarrow_d N(0,1) \label{MildLimit2P}  \ \   
   \end{align}
\end{itemize}
\end{theorem}

\begin{remark}\label{REM_Mildly} 

\begin{enumerate}
\item Results in \eqref{MildLimit2A} and \eqref{MildLimit2P} are analogue; we conjecture that a result analogue to \eqref{MildLimit1}  also exists when $\widehat{\sigma}_A$ is used and $M/T^{\alpha} \rightarrow \infty$.

\item 
Rewriting 
 $ DM_P= 
  \sqrt{T}\frac{\overline{d}-\mu}
   {\widehat{\sigma}_P} 
   + \sqrt{T}\frac{\mu}
   {\widehat{\sigma}_P} $
for $m$ as in \eqref{MildLimit1}, the power depends on the drift 
$    \sqrt{T}\frac{\mu}
   {\widehat{\sigma}_P}
   =O(T^{1/2-\alpha})$. Therefore, the DM test still has power in cases of moderate deviations from the unit root, when $\alpha<1/2$. However, from this representation, it is immediate to see that, as the drift is of order $T^{1/2-\alpha}$, the power decreases as $\alpha \rightarrow 1/2$. This result means that it is progressively more difficult to detect forecast inaccuracy as the dependence increases, even well within the weak dependence region.
\item When
   $m T^{\alpha-1} \rightarrow \infty$ (or $M/T^\alpha \rightarrow 0$), the DM statistics diverges even under $H_0$, thus resulting again in spurious significance. 
\item Condition $m T^{\alpha-1} \rightarrow 0$ in \eqref{MildLimit1} is not binding when $\alpha=0$ but it is very strong as $\alpha \rightarrow 1$. Thus, results in \eqref{MildLimit1} and \eqref{MildLimit2P} are intermediate between the weakly dependent $|\rho_T|=|\rho|<1$ case and the unit root case. Taken together, they suggest that for large values of $m$ the DM test will give spurious significance in finite samples since $\rho$ is close to 1, and this problem is more relevant the closer $\rho$ is to unity, relative to the sample size, and the larger the bandwidth $m$. 
\end{enumerate}
\end{remark}

\section{Monte Carlo results}\label{sec:MC}

In this section, we investigate the properties of the DM statistic in the neighbourhood of unity in a Monte Carlo exercise. We consider the DGP
\begin{equation}\label{DGP}
y_t=\alpha+\beta x_{t-1}+u_{t}
\end{equation}
where
\begin{align*}
& x_t=\phi x_{t-1}+\varepsilon_{t}, \ |\phi|<1, \ \varepsilon_{t} \sim N.i.d.(0,\sigma^2_\varepsilon) \\
& u_{t} \sim N.i.d.(0,\sigma^2_u)
\end{align*}
and $u_{t}$ independent from $\varepsilon_s$ for all $s,t$. 

Notice that we have assumed $E(x_t)=0$, this is without loss of generality because in \eqref{DGP} we could use deviations from the mean, $y_t=(\alpha+\beta E(x_{t-1}))+\beta (x_{t-1}-E(x_{t-1}))+u_{t}$. Finally, we also set $\beta=1$, again without loss of generality. 

We consider two forecasting strategies:
\begin{align}
&y_{1,t}=\widehat{\beta}x_{t-1} \ \text{  where  } \widehat{\beta}=Plim_{R\rightarrow\infty}\frac{\sum_{s=t-R}^{t-1}y_s x_{s-1}}{\sum_{s=t-R}^{t-1} x_{s-1}^2}\label{for1} \\
& y_{2,t}=\widetilde{y} \ \text{  where  } \widetilde{y}=Plim_{R\rightarrow\infty}\frac{1}{R}\sum_{s=t-R}^{t-1}y_s\label{for2}
\end{align}
so $\widehat{\beta}=1$ (in general it would be $\beta$ ) and $\widetilde{y}=\alpha$, and
\begin{align*}
&e_{1,t}=y_t-y_{1,t}= \alpha+u_t \\
& e_{2,t}=y_t-y_{2,t}=
x_{t-1}+u_t  \\
& d_t=(e_{1,t}^2- e_{2,t}^2)=
\alpha^2+2 \alpha u_t - x_{t-1}^2 - 2 x_{t-1} u_t
\end{align*}
and
\begin{equation*}
E(d_t)=(e_{1,t}^2-e_{2,t}^2)=
\alpha^2- \frac{1}{1-\phi^2} \sigma_\varepsilon^2
\end{equation*}
so, for
\begin{equation}\label{alpha}
\alpha=\sqrt{\frac{1}{1-\phi^2} \sigma_\varepsilon^2}+\delta
\end{equation}
then $E(d_t)=0$  when $\delta=0$. Notice that as $|\phi|<1$ then both $x_t$, $y_t$ and $d_t$  are mixing with sufficient rate, $E(d_s^2)<\infty$ (in view of the Gaussianity) and the long run variance exists. \bigskip \\
\textbf{Remark}
\textit{From Lemma 1 of \citeasnoun{dittmann2002properties}, $x_{t-1}^2$ is AR(1) with coefficient $\phi^2$; $\alpha u_t$ is an independent process and $x_{t-1} u_t$ is Martingale difference. Thus $d_t$ is like AR(1) plus noise.} 

As in any realistic situation the sample size $T$ is given, whether the standard normal limit or Theorem \ref{DF-standard} is a better approximation depends on the relative interplay between $\phi$ and $T$ (by the same token, the Lemma in \citeasnoun{dittmann2002properties} should also be seen  as an approximation, when $\phi$ is close to 1, relative to $T$). For values of $\phi$ close to 1 (relative to $T$) the limit \eqref{Lemma1bP} would be a better approximation for the sample average, and Theorem \ref{DF-standard} is a more appropriate guideline; conversely, \eqref{standardCLT} and the standard normal limit for the DM statistic should be a more reliable guideline when $\phi$ is not close to 1, relative to $T$. Thus, the same value of $\phi$ could generate spurious rejections or not depending on the sample size. In this Monte Carlo study, we  thus consider a range of values for $\phi$ and $T$ to assess the interplay between these two key elements.

We consider two sample sizes,  $T=50$ and $T=100$, and  a range of bandwidths spanning $m=1$ to $m=\left\lfloor T^{2/3} \right\rfloor$ when the average periodogram is used, and a range of bandwidths spanning from $M=\left\lfloor T^{1/4} \right\rfloor$ to $M=T$ when the weighted autocovariance estimator is used. For each experiment, we run $10,000$ repetitions, and we compute the empirical frequencies of rejections of the two-sided version of the test, i.e. we compare the $|DM|$  statistic against the appropriate $5 \%$ critical value from the $t_{2m}$ or the $\Phi_A(b)$ distribution, respectively, the latter as in \citeasnoun{kiefer2005new}. We always use these critical values since they yield better size properties, as discussed, for example, in \citeasnoun{lazarus2018har} or in \citeasnoun{coroneo2020comparing}.

We consider $\sigma_{\varepsilon}^2 = 1$, $\sigma_u^2 = 1$, and  a range of values for $\phi$; $\alpha$ is as in \eqref{alpha} for two values of $\delta$. We set $\delta=0$ to observe the effects on the empirical size; to observe the effects on power we set $\delta=
-\sqrt{\sigma^2_\varepsilon/(1-\phi^2)}+\sqrt{\sigma^2_\varepsilon/(1-\phi^2)-1)}$ as this yields $E(d_t)=-1$: with this choice we can observe how the power changes as the persistence increases, for the same deviation from the null hypothesis.

\begin{table}[t!]
  \centering
  \caption{Empirical null rejection frequencies - weighted autocovariances}   \label{tab:MCA}%
    \begin{tabular}{cccccccccc}
    \\
    \multicolumn{10}{c}{$
    E(d_t)=0$, $T=50$}\\
    \hline \hline
     \multicolumn{1}{c|}{\backslashbox{$M$}{$\phi$}} & 0     & 0.5   & 0.75  & 0.8   & 0.85  & 0.9   & 0.95  & 0.975 & 0.99  \\
          \hline
    $\left\lfloor T^{1/4} \right\rfloor $ & 0.055    & 0.069    & 0.134    & 0.176    & 0.237    & 0.353    & 0.556    & 0.706    & 0.827 \\
    $\left\lfloor T^{2/9} \right\rfloor $   & 0.055    & 0.069    & 0.134    & 0.176    & 0.237    & 0.353    & 0.556    & 0.706    & 0.827 \\
    $\left\lfloor T^{1/3} \right\rfloor $   & 0.054    & 0.062    & 0.118    & 0.152    & 0.204    & 0.307    & 0.500    & 0.658    & 0.798 \\
    $\left\lfloor T^{1/2} \right\rfloor $   & 0.052    & 0.056    & 0.092    & 0.118    & 0.155    & 0.229    & 0.380    & 0.539    & 0.710 \\
    {T}        & 0.051    & 0.056    & 0.083    & 0.100    & 0.128    & 0.178    & 0.284    & 0.401    & 0.545 \\
          \hline
     \\
    \multicolumn{10}{c}{
   $E(d_t)=-1$, $T=50$}\\
    \hline \hline
     \multicolumn{1}{c|}{\backslashbox{$M$}{$\phi$}} & 0     & 0.5   & 0.75  & 0.8   & 0.85  & 0.9   & 0.95  & 0.975 & 0.99   \\
          \hline
    $\left\lfloor T^{1/4} \right\rfloor $ & 0.885 & 0.583 & 0.244 & 0.208 & 0.217 & 0.309 & 0.530 & 0.696 & 0.826 \\
    $\left\lfloor T^{2/9} \right\rfloor $  & 0.885 & 0.583 & 0.244 & 0.208 & 0.217 & 0.309 & 0.530 & 0.696 & 0.826 \\
    $\left\lfloor T^{1/3} \right\rfloor $ & 0.873 & 0.554 & 0.197 & 0.157 & 0.169 & 0.255 & 0.471 & 0.645 & 0.794 \\
    $\left\lfloor T^{1/2} \right\rfloor $ & 0.817 & 0.482 & 0.120 & 0.085 & 0.093 & 0.157 & 0.343 & 0.525 & 0.706 \\
    $T$     & 0.654 & 0.379 & 0.099 & 0.070 & 0.072 & 0.117 & 0.247 & 0.386 & 0.541 \\

            \hline
     \\
    \multicolumn{10}{c}{$
    E(d_t)=0$, $T=100$}\\
    \hline \hline
     \multicolumn{1}{c|}{\backslashbox{$M$}{$\phi$}} & 0     & 0.5   & 0.75  & 0.8   & 0.85  & 0.9   & 0.95  & 0.975 & 0.99   \\
          \hline
  $\left\lfloor T^{1/4} \right\rfloor $ & 0.047    & 0.056    & 0.108    & 0.140    & 0.194    & 0.294    & 0.489    & 0.643    & 0.788 \\
    $\left\lfloor T^{2/9} \right\rfloor $   & 0.048    & 0.060    & 0.125    & 0.165    & 0.232    & 0.350    & 0.547    & 0.694    & 0.823 \\
    $\left\lfloor T^{1/3} \right\rfloor $   & 0.047    & 0.053    & 0.098    & 0.126    & 0.171    & 0.259    & 0.447    & 0.603    & 0.760 \\
   $\left\lfloor T^{1/2} \right\rfloor $  & 0.044    & 0.047    & 0.078    & 0.095    & 0.126    & 0.180    & 0.309    & 0.458    & 0.645 \\
    $T$        & 0.046    & 0.047    & 0.066    & 0.079    & 0.100    & 0.136    & 0.214    & 0.310    & 0.441 \\
            \hline
     \\
    \multicolumn{10}{c}{
    $E(d_t)=-1$, $T=100$}\\
    \hline \hline
     \multicolumn{1}{c|}{\backslashbox{$M$}{$\phi$}} & 0     & 0.5   & 0.75  & 0.8   & 0.85  & 0.9   & 0.95  & 0.975 & 0.99   \\
          \hline

    $\left\lfloor T^{1/4} \right\rfloor $ & 0.997 & 0.880 & 0.376 & 0.275 & 0.224 & 0.263 & 0.460 & 0.634 & 0.787 \\
    $\left\lfloor T^{2/9} \right\rfloor $ & 0.997 & 0.892 & 0.430 & 0.328 & 0.280 & 0.322 & 0.523 & 0.680 & 0.820 \\
    $\left\lfloor T^{1/3} \right\rfloor $ & 0.995 & 0.871 & 0.336 & 0.233 & 0.183 & 0.221 & 0.412 & 0.591 & 0.755 \\
    $\left\lfloor T^{1/2} \right\rfloor $ & 0.990 & 0.826 & 0.237 & 0.135 & 0.087 & 0.117 & 0.266 & 0.441 & 0.641 \\
    $T$     & 0.884 & 0.629 & 0.180 & 0.104 & 0.067 & 0.084 & 0.176 & 0.293 & 0.439 \\
\hline
\\
\multicolumn{10}{p{12cm}}{Note: empirical null rejection frequencies for the DM test with the $DM_A$ statistic and 
fixed-$b$ critical values. The data generating process is in equations~\eqref{DGP} and \eqref{alpha}, with $E(d_t)=0$ for the size exercise, and $E(d_t)=-1$ for the power study. The sample size is 50 and 100; 10,000 replications.}
    \end{tabular}
\end{table}%

\begin{table}[t!]
  \centering
  \caption{Empirical null rejection frequencies - weighted periodogram}   \label{tab:MCP}%
    \begin{tabular}{cccccccccc}
    \\
    \multicolumn{10}{c}{$E(d_t)=0$, $T=50$}\\
    \hline \hline
     \multicolumn{1}{c|}{\backslashbox{$m$}{$\phi$}} & 0     & 0.5   & 0.75  & 0.8   & 0.85  & 0.9   & 0.95  & 0.975 & 0.99  \\
          \hline
    {1}      & 0.050    & 0.048    & 0.058    & 0.068    & 0.082    & 0.112    & 0.184    & 0.279    & 0.420 \\
    $\left\lfloor T^{1/4} \right\rfloor $  & 0.052    & 0.054    & 0.075    & 0.088    & 0.114    & 0.160    & 0.270    & 0.407    & 0.594 \\
   $\left\lfloor T^{1/3} \right\rfloor $   & 0.054    & 0.055    & 0.077    & 0.098    & 0.130    & 0.191    & 0.324    & 0.484    & 0.669 \\
   $\left\lfloor T^{1/2} \right\rfloor $   & 0.053    & 0.060    & 0.103    & 0.133    & 0.181    & 0.274    & 0.463    & 0.628    & 0.779 \\
   $\left\lfloor T^{2/3} \right\rfloor $   & 0.056    & 0.067    & 0.133    & 0.173    & 0.234    & 0.351    & 0.558    & 0.709    & 0.830 \\
      \hline
     \\
    \multicolumn{10}{c}{$E(d_t)=-1$, $T=50$}\\
    \hline \hline
     \multicolumn{1}{c|}{\backslashbox{$m$}{$\phi$}} & 0     & 0.5   & 0.75  & 0.8   & 0.85  & 0.9   & 0.95  & 0.975 & 0.99   \\
          \hline
    {1}  & 0.356 & 0.201 & 0.059 & 0.047 & 0.048 & 0.073 & 0.157 & 0.264 & 0.413 \\
    $\left\lfloor T^{1/4} \right\rfloor $ & 0.610 & 0.321 & 0.077 & 0.052 & 0.056 & 0.100 & 0.233 & 0.392 & 0.590 \\
    $\left\lfloor T^{1/3} \right\rfloor $ & 0.728 & 0.401 & 0.085 & 0.057 & 0.066 & 0.122 & 0.286 & 0.471 & 0.666 \\
    $\left\lfloor T^{1/2} \right\rfloor $ & 0.847 & 0.506 & 0.153 & 0.123 & 0.132 & 0.215 & 0.432 & 0.616 & 0.776 \\
    $\left\lfloor T^{2/3} \right\rfloor $ & 0.878 & 0.571 & 0.238 & 0.204 & 0.215 & 0.310 & 0.532 & 0.698 & 0.827 \\

           \hline
     \\
    \multicolumn{10}{c}{$E(d_t)=0$, $T=100$}\\
    \hline \hline
     \multicolumn{1}{c|}{\backslashbox{$m$}{$\phi$}} & 0     & 0.5   & 0.75  & 0.8   & 0.85  & 0.9   & 0.95  & 0.975 & 0.99    \\
          \hline
    {1}      & 0.049    & 0.048    & 0.053    & 0.059    & 0.071    & 0.092    & 0.135    & 0.203    & 0.320 \\
     $\left\lfloor T^{1/4} \right\rfloor $    & 0.045    & 0.045    & 0.064    & 0.078    & 0.099    & 0.140    & 0.226    & 0.350    & 0.544 \\
    $\left\lfloor T^{1/3} \right\rfloor $   & 0.046    & 0.049    & 0.067    & 0.086    & 0.108    & 0.151    & 0.255    & 0.399    & 0.597 \\
   $\left\lfloor T^{1/2} \right\rfloor $    & 0.048    & 0.051    & 0.084    & 0.105    & 0.146    & 0.215    & 0.397    & 0.560    & 0.730 \\
   $\left\lfloor T^{2/3} \right\rfloor $    & 0.048    & 0.056    & 0.111    & 0.146    & 0.209    & 0.320    & 0.519    & 0.671    & 0.810 \\
       \hline
     \\
    \multicolumn{10}{c}{$E(d_t)=-1$, $T=100$}\\
    \hline \hline
     \multicolumn{1}{c|}{\backslashbox{$m$}{$\phi$}} & 0     & 0.5   & 0.75  & 0.8   & 0.85  & 0.9   & 0.95  & 0.975 & 0.99   \\
          \hline
     {1}   & 0.582 & 0.349 & 0.109 & 0.075 & 0.053 & 0.058 & 0.107 & 0.188 & 0.315 \\
    $\left\lfloor T^{1/4} \right\rfloor $ & 0.959 & 0.698 & 0.161 & 0.085 & 0.052 & 0.073 & 0.182 & 0.332 & 0.539 \\
    $\left\lfloor T^{1/3} \right\rfloor $ & 0.977 & 0.758 & 0.181 & 0.097 & 0.057 & 0.085 & 0.213 & 0.382 & 0.593 \\
    $\left\lfloor T^{1/2} \right\rfloor $ & 0.993 & 0.848 & 0.274 & 0.178 & 0.134 & 0.175 & 0.359 & 0.547 & 0.727 \\
    $\left\lfloor T^{2/3} \right\rfloor $ & 0.996 & 0.878 & 0.387 & 0.293 & 0.245 & 0.288 & 0.490 & 0.660 & 0.807 \\
\hline
\\
\multicolumn{10}{p{12cm}}{Note: empirical null rejection frequencies for the DM test with the $DM_P$ statistic and  fixed-$m$ critical values. The data generating process is in equations\eqref{DGP} and \eqref{alpha}, with $E(d_t)=0$ for the size exercise, and $E(d_t) =-1$ for the power study. The sample size is 50 and 100; 10,000 replications.}
    \end{tabular}
\end{table}%

We report in Tables~\ref{tab:MCA}-\ref{tab:MCP}
the simulation results using the weighted autocovariance and the weighted periodogram estimates of the long-run variance, respectively.
Results confirm the findings in Section~\ref{sec_theory}. In particular:
\begin{enumerate}
    \item The empirical size is correct for $\phi=0$, but (for given $T$) it deteriorates as we move closer to $\phi=1$ and as $M$ is smaller or $m$ is larger.  
    \item The empirical power drops as we move closer to $\phi=1$ and as $M$ is smaller or $m$ is larger, in the sense that the presence of $E(d_t) \neq 0$ does not affect much the number of rejections of the null hypothesis in those cases. 
\end{enumerate}

These results  support the key conclusions that we derived in Section \ref{sec_theory}. In particular, in the size exercise, the distortion increases with $\phi$ and with the bandwidth $m$ (and decreases with $M$). In the power exercise, the power drops as $\phi$ increases from 0 to 0.85, but notice that for large values of $\phi$ this power is in fact fictitious, in the sense that it rather reflects the spurious size distortion that we observed under the null. \label{R2.5} \Copy{R2.5t}{Automatic procedures to select the bandwidth, as in \citeasnoun{delgado1996optimal}, \citeasnoun{robinson1983review} or in \citeasnoun{newey1994automatic} would not solve these problems, although the fact that smaller $m$s (larger $M$s) are automatically selected as the autocorrelation in the loss function increases, would at least avoid the most adverse effects. }
\clearpage

\section{Empirical application}\label{sec:emp}
To illustrate the problems associated with the DM test when there is dependence in the loss differential, in this section we present the case in which a forecast for inflation  with  weakly dependent forecast errors is compared to two strongly dependent naive benchmarks. In particular, we consider quarterly predictions for the inflation rate in the Euro Area from a standard AR(1) model, as in  \citeasnoun{forni2003financial} and \citeasnoun{marcellino2003macroeconomic}. As for the benchmarks, we consider a constant 2\% prediction (that represents the inflation target in the Euro Area) and a rolling  average (RA) prediction. 

We use data on the Harmonized Index of Consumer Prices from the FRED database, and we compute quarterly year-on-year inflation rates from 1996.Q1 to 2020.Q4. \label{R1.4} \Copy{R1.4t}{Given that our objective is to compare forecasting methods as opposed to forecasting models, we consider the case of non-vanishing estimation uncertainty, as in \citeasnoun{giacomini2006tests}, and estimate all coefficients and rolling averages using a rolling window of 10 years.} We compute predictions for horizons from 1 quarter to 8 quarters-ahead, and we evaluate them on the period from 2010.Q1 to 2020.Q4 (44 observations) using a quadratic loss function.

The series of inflation and the forecasts for selected forecast horizons for the AR(1) model are shown in Figure~\ref{fig:Infl}, along with the $2\%$ and the rolling average benchmarks. A visual inspection of the plots immediately suggests that the forecast from the AR(1) model is  clearly superior to the $2\%$ benchmark for the 2-quarter horizon, but, not surprisingly, the superior performance of the forecast from the AR(1) model is eroded as longer forecasting horizons are considered. Additional plots for the realised forecast errors, the realised losses and the realised loss differentials are reported in Appendix~\ref{app_plots}.

\begin{figure}[H]
\caption{Realised inflation and forecasts \bigskip} 
\includegraphics[trim={1.5cm 8cm 1cm 8cm},clip,scale=0.9]{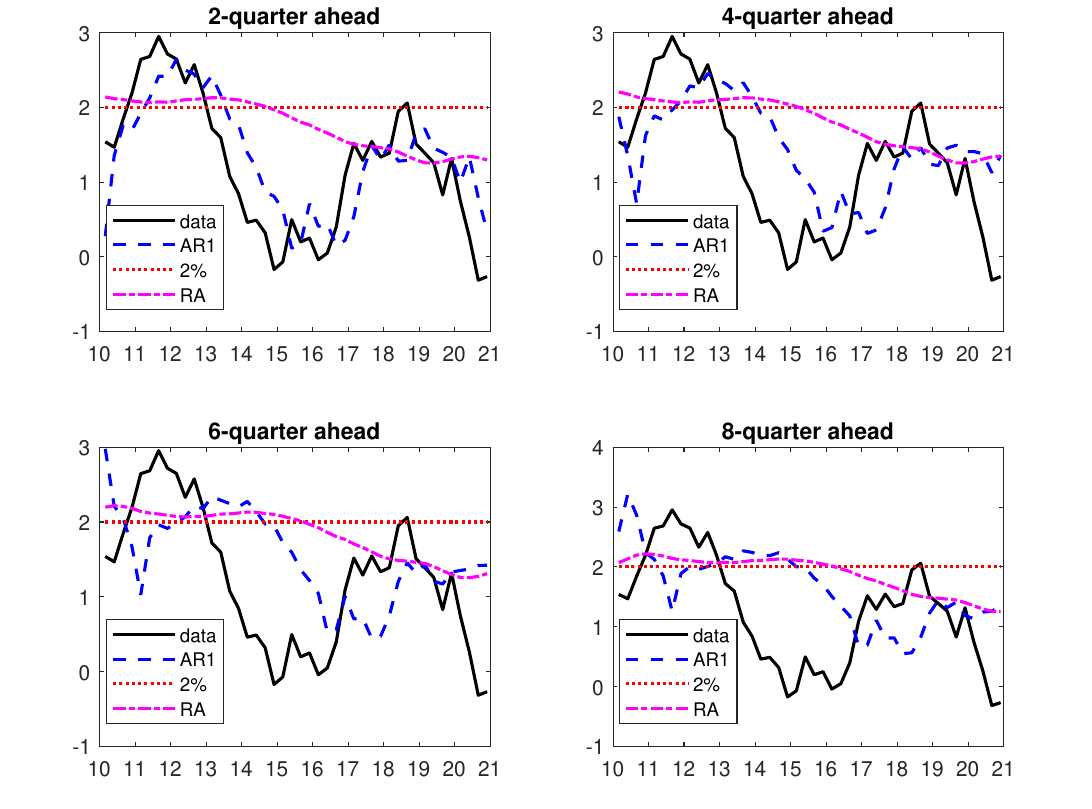}
\label{fig:Infl}
Note: realised inflation (data), AR(1)  forecasts, along with the $2\%$  ($2\%$) and the rolling average (RA) benchmark forecasts for forecasting horizons 2, 4, 6 and 8 quarters.
\end{figure}

\begin{table}[t!]
  \centering
  \caption{Summary statistics of forecast errors}
    \begin{tabular}{crrrrrrrl}
    \\
    \multicolumn{8}{c}{AR(1) forecast}\\
    \midrule
    Horizon & \multicolumn{1}{l}{Mean} & \multicolumn{1}{l}{Median} & \multicolumn{1}{l}{Std} & \multicolumn{1}{l}{AC1} & \multicolumn{1}{l}{AC2} & \multicolumn{1}{l}{AC3} & \multicolumn{1}{l}{AC4}& \multicolumn{1}{l}{ADF}\\
    \midrule
    1        & -0.059   & -0.077   & 0.367    & 0.284    & 0.264    & 0.182    & 0.179    & -4.893$^{**}$ \\
    2        & -0.100   & -0.092   & 0.590    & 0.690    & 0.384    & 0.350    & 0.313    & -1.823 \\
    3        & -0.148   & -0.135   & 0.724    & 0.817    & 0.627    & 0.433    & 0.397    & -1.892 \\
    4        & -0.204   & -0.134   & 0.827    & 0.860    & 0.697    & 0.569    & 0.424    & -1.587 \\
    5        & -0.252   & -0.314   & 0.904    & 0.882    & 0.733    & 0.603    & 0.456    & -2.222 \\
    6        & -0.317   & -0.182   & 0.983    & 0.898    & 0.760    & 0.577    & 0.390    & -2.417 \\
    7        & -0.378   & -0.288   & 1.046    & 0.907    & 0.738    & 0.551    & 0.354    & -2.009 \\
    8        & -0.419   & -0.420   & 1.078    & 0.881    & 0.717    & 0.535    & 0.366    & -2.014 \\

   \midrule
    \\
    \multicolumn{8}{c}{RA forecast}\\
    \midrule
    Horizon & \multicolumn{1}{l}{Mean} & \multicolumn{1}{l}{Median} & \multicolumn{1}{l}{Std} & \multicolumn{1}{l}{AC1} & \multicolumn{1}{l}{AC2}& \multicolumn{1}{l}{AC3} & \multicolumn{1}{l}{AC4}& \multicolumn{1}{l}{ADF}\\
    \midrule
    1        & -0.505   & -0.411   & 0.855    & 0.912    & 0.787    & 0.639    & 0.491    & -1.116 \\
    2        & -0.525   & -0.428   & 0.871    & 0.914    & 0.789    & 0.641    & 0.493    & -1.085 \\
    3        & -0.546   & -0.413   & 0.882    & 0.916    & 0.792    & 0.644    & 0.497    & -1.074 \\
    4        & -0.566   & -0.409   & 0.890    & 0.917    & 0.794    & 0.648    & 0.503    & -1.065 \\
    5        & -0.585   & -0.427   & 0.895    & 0.917    & 0.795    & 0.653    & 0.509    & -1.083 \\
    6        & -0.605   & -0.438   & 0.896    & 0.919    & 0.799    & 0.658    & 0.517    & -1.091 \\
    7        & -0.624   & -0.454   & 0.895    & 0.919    & 0.802    & 0.665    & 0.526    & -1.108 \\
    8        & -0.642   & -0.475   & 0.893    & 0.920    & 0.806    & 0.673    & 0.536    & -1.123 \\

    \midrule
    \\
    \multicolumn{8}{c}{2\% forecast}\\
    \midrule
    Horizon & \multicolumn{1}{l}{Mean} & \multicolumn{1}{l}{Median} & \multicolumn{1}{l}{Std} & \multicolumn{1}{l}{AC1} & \multicolumn{1}{l}{AC2}& \multicolumn{1}{l}{AC3} & \multicolumn{1}{l}{AC4}& \multicolumn{1}{l}{ADF}\\
    \midrule
    1-8       & -0.765   & -0.673   & 0.923    & 0.926    & 0.821    & 0.695    & 0.569    & -0.732 \\
    \bottomrule
    \multicolumn{9}{p{12cm}}{Note:  summary statistics for forecast errors from AR(1) predictions (top panel), rolling  average (RA) predictions (middle panel) and constant 2\% predictions. Forecast horizons are in quarters and forecast errors are defined as the realised value minus the prediction. ADF refers to the augmented Dickey–Fuller test (with intercept and lag order selected using the BIC criterion). $^*$ and $^{**}$ denote significance at 10\% and 5\% level.}
    \end{tabular}%
  \label{tab:summarye}%
\end{table}%

\begin{table}[t!]
  \centering
  \caption{Summary statistics of realised losses}
    \begin{tabular}{crrrrrrrl}
    \\
    \multicolumn{8}{c}{AR(1) forecast}\\
    \midrule
    Horizon & \multicolumn{1}{l}{Mean} & \multicolumn{1}{l}{Median} & \multicolumn{1}{l}{Std} & \multicolumn{1}{l}{AC1} & \multicolumn{1}{l}{AC2} & \multicolumn{1}{l}{AC3} & \multicolumn{1}{l}{AC4}& \multicolumn{1}{l}{ADF}\\
    \midrule
     1     & 0.135 & 0.092 & 0.128 & -0.092 & 0.121 & -0.030 & -0.250 & -6.938$^{**}$ \\
    2     & 0.350 & 0.230 & 0.394 & 0.301 & -0.145 & -0.073 & 0.023 & -4.735$^{**}$ \\
    3     & 0.533 & 0.317 & 0.596 & 0.462 & 0.181 & 0.028 & 0.004 & -3.820$^{**}$ \\
    4     & 0.710 & 0.404 & 0.764 & 0.693 & 0.416 & 0.213 & -0.016 & -2.466 \\
    5     & 0.862 & 0.486 & 0.927 & 0.738 & 0.470 & 0.292 & 0.112 & -2.429 \\
    6     & 1.045 & 0.600 & 1.133 & 0.706 & 0.459 & 0.303 & 0.182 & -2.208 \\
    7     & 1.213 & 0.732 & 1.240 & 0.776 & 0.472 & 0.308 & 0.200 & -2.094 \\
    8     & 1.312 & 0.700 & 1.323 & 0.728 & 0.486 & 0.329 & 0.247 & -2.163 \\   
   \midrule
    \\
    \multicolumn{8}{c}{RA forecast}\\
    \midrule
    Horizon & \multicolumn{1}{l}{Mean} & \multicolumn{1}{l}{Median} & \multicolumn{1}{l}{Std} & \multicolumn{1}{l}{AC1} & \multicolumn{1}{l}{AC2}& \multicolumn{1}{l}{AC3} & \multicolumn{1}{l}{AC4}& \multicolumn{1}{l}{ADF}\\
    \midrule
    1     & 0.970 & 0.363 & 1.172 & 0.849 & 0.663 & 0.519 & 0.400 & -1.688 \\
    2     & 1.017 & 0.357 & 1.241 & 0.858 & 0.672 & 0.527 & 0.410 & -1.610 \\
    3     & 1.058 & 0.353 & 1.304 & 0.863 & 0.682 & 0.538 & 0.420 & -1.579 \\
    4     & 1.094 & 0.341 & 1.357 & 0.869 & 0.693 & 0.553 & 0.433 & -1.541 \\
    5     & 1.125 & 0.349 & 1.401 & 0.874 & 0.705 & 0.567 & 0.449 & -1.526 \\
    6     & 1.150 & 0.336 & 1.443 & 0.877 & 0.713 & 0.579 & 0.459 & -1.546 \\
    7     & 1.171 & 0.329 & 1.478 & 0.881 & 0.721 & 0.587 & 0.469 & -1.555 \\
    8     & 1.191 & 0.332 & 1.510 & 0.883 & 0.727 & 0.592 & 0.475 & -1.564 \\
    \midrule
    \\
    \multicolumn{8}{c}{2\% forecast}\\
    \midrule
    Horizon & \multicolumn{1}{l}{Mean} & \multicolumn{1}{l}{Median} & \multicolumn{1}{l}{Std} & \multicolumn{1}{l}{AC1} & \multicolumn{1}{l}{AC2}& \multicolumn{1}{l}{AC3} & \multicolumn{1}{l}{AC4}& \multicolumn{1}{l}{ADF}\\
    \midrule
 1-8    & 1.418 & 0.509 & 1.593 & 0.862 & 0.654 & 0.461 & 0.338 & -1.581 \\
    \bottomrule
    \multicolumn{9}{p{12cm}}{Note:  summary statistics for realised losses from AR(1) predictions (top panel), rolling  average (RA) predictions (middle panel) and constant 2\% predictions. Forecast horizons are in quarters and forecast errors are defined as the realised value minus the prediction. ADF refers to the augmented Dickey–Fuller test (with intercept and lag order selected using the BIC criterion). $^*$ and $^{**}$ denote significance at 10\% and 5\% level.}
    \end{tabular}%
  \label{tab:summaryL}%
\end{table}%

In Table~\ref{tab:summarye}, we report summary statistics for the forecast errors (defined as the realised value minus the prediction) for the AR(1) and the two benchmark predictions. The forecast errors are all negative on average, implying that in this period inflation in the Euro Area has been lower than predicted by the AR(1) and the benchmarks. This result is not generated by a few large negative errors, as all the median forecast errors are also negative.

The average and median forecast errors for the AR(1) increase (in absolute value) with the forecast horizon,  but they remain lower than the ones of the two benchmarks for all the forecasting horizons. The standard deviations of the AR(1) forecast errors also increase with the forecasting horizon, and they are  smaller than the ones of the two benchmarks for forecasting horizons up to 4 quarters. 
Finally, we also present the autocorrelation structure and the ADF tests for the forecast errors (we estimated the model with the intercept, with lags selected by means of the BIC). Especially at the lowest horizons, the autocorrelations of the errors from the AR(1) forecasts declines fairly quickly, in comparison with the autocorrelations of the benchmarks. Despite this fact, the ADF test fails to reject the null hypothesis in all the cases, except for the one period horizon.  We interpret this as a situation of low power of the ADF test, and therefore as evidence of persistence, but possibly not a unit root. We verify this interpretation by 
analysing the properties of the realised forecast losses reported in Table~\ref{tab:summaryL}. The average realised losses associated to the AR(1) forecast are lower, at least for forecasts up to six quarters, and less dispersed than the ones of the two benchmarks, so they are, in this sense, more precise. Moreover, the losses from the AR(1) predictions are not very correlated for short forecasting horizons. As we increase the forecasting horizon the dependence increases, but the autocorrelations still decay reasonably quickly.  On the other hand, the two benchmarks display large and persistent  autocorrelations in their realised forecast losses at all forecasting horizons. We further investigate the dependence in the realised losses using the ADF test: the difference in the persistence that we observed in the sample autocorrelations of the realised losses is confirmed by the outcome of the ADF test, where the unit root hypothesis is rejected only for the forecasts from the AR(1) model (and only for short horizons).

Overall, these results suggest that the AR(1) model should be more precise for short-term forecasts, but this superiority could be masked empirically by the excessive dependence in the benchmarks. On the other hand, the AR(1) does not seem to produce more precise forecasts than the benchmarks at longer horizons, and the outcomes of the unit roots tests should be interpreted as a warning that any potential statistical significant difference may be spurious.

\begin{table}[btp]
  \centering
  \caption{Summary statistics loss differential}
    \begin{tabular}{crrrrrrrl}
    \\
    \multicolumn{9}{c}{Benchmark: RA}\\
    \midrule
    \multicolumn{1}{l}{Horizon} & \multicolumn{1}{l}{Mean} & \multicolumn{1}{l}{Median} & \multicolumn{1}{l}{Std} & \multicolumn{1}{l}{AC1} & \multicolumn{1}{l}{AC2} &
    \multicolumn{1}{l}{AC3} & \multicolumn{1}{l}{AC4} &\multicolumn{1}{l}{ADF} \\
    \midrule
    1        & 0.835    & 0.261    & 1.152    & 0.852    & 0.652    & 0.499    & 0.385    & -1.624 \\
    2        & 0.667    & 0.067    & 1.181    & 0.826    & 0.647    & 0.504    & 0.382    & -1.855 \\
    3        & 0.525    & -0.001   & 1.127    & 0.848    & 0.675    & 0.501    & 0.348    & -2.172 \\
    4        & 0.384    & -0.003   & 1.093    & 0.837    & 0.675    & 0.479    & 0.283    & -2.581 \\
    5        & 0.263    & -0.020   & 1.046    & 0.843    & 0.644    & 0.440    & 0.233    & -3.047$^{**}$\\
    6        & 0.105    & -0.068   & 1.015    & 0.770    & 0.583    & 0.400    & 0.283    & -2.128 \\
    7        & -0.041   & -0.074   & 0.954    & 0.797    & 0.511    & 0.370    & 0.299    & -1.999 \\
    8        & -0.121   & -0.076   & 0.865    & 0.651    & 0.386    & 0.230    & 0.195    & -3.678$^{**}$ \\
    \midrule
    \\
    \multicolumn{9}{c}{Benchmark: 2\%}\\
    \midrule
    \multicolumn{1}{l}{Horizon} & \multicolumn{1}{l}{Mean} & \multicolumn{1}{l}{Median} & \multicolumn{1}{l}{Std} & \multicolumn{1}{l}{AC1} & \multicolumn{1}{l}{AC2}&
    \multicolumn{1}{l}{AC3} & \multicolumn{1}{l}{AC4} &\multicolumn{1}{l}{ADF} \\
    \midrule
    1        & 1.283    & 0.449    & 1.573    & 0.862    & 0.635    & 0.437    & 0.320    & -1.664 \\
    2        & 1.068    & 0.433    & 1.530    & 0.834    & 0.626    & 0.443    & 0.305    & -1.353 \\
    3        & 0.885    & 0.335    & 1.389    & 0.851    & 0.643    & 0.440    & 0.269    & -1.775 \\
    4        & 0.708    & 0.305    & 1.306    & 0.834    & 0.640    & 0.410    & 0.225    & -2.601 \\
    5        & 0.556    & 0.121    & 1.255    & 0.846    & 0.634    & 0.414    & 0.191    & -2.556 \\
    6        & 0.373    & 0.039    & 1.208    & 0.794    & 0.602    & 0.416    & 0.265    & -2.290 \\
    7        & 0.206    & 0.026    & 1.192    & 0.819    & 0.561    & 0.395    & 0.275    & -2.517 \\
    8        & 0.106    & 0.009    & 1.161    & 0.736    & 0.474    & 0.285    & 0.185    & -2.626$^{*}$ \\
    \bottomrule
    \multicolumn{9}{p{12cm}}{Note:  summary statistics for the AR(1) loss differential with respect to rolling average (RA) predictions (top panel) and constant 2\% predictions (bottom panel). The loss function is quadratic and the loss differential is computed as the loss of the benchmark minus the loss of the AR(1). ADF refers to the augmented Dickey–Fuller test (with intercept and lag order selected using the BIC). $^*$ and $^{**}$ denote significance at 10\% and 5\% level.}
    \end{tabular}%
  \label{tab:summary_loss}%
\end{table}%

Summary statistics of the loss differentials, computed
as the loss of the benchmark minus the loss of the AR(1), reported in Table~\ref{tab:summary_loss},  show that at short  horizons loss differentials are positive, indicating that AR(1) predictions  may be more accurate than the benchmarks. As the forecasting horizon increases, the average loss differential decreases. In particular, for the RA benchmark, it becomes negative, so that at 7 and 8 quarters ahead RA predictions  might be more accurate than AR(1) predictions. 

However, the table also shows that the loss differentials are characterised by relevant autocorrelations, even at short forecasting horizons: the properties of the loss differentials are therefore heavily affected by the benchmark considered, and even with a forecast with weakly dependent loss it is possible to have strong autocorrelation of the loss differential.  In the last column of the table, we report the augmented Dickey–Fuller test statistic, which clearly indicates that even at short forecasting horizons the null of unit root of the loss differential cannot be rejected. With these levels of dependence, the DM test statistic is going to be subject to the drawbacks described in Section~\ref{sec_theory}.
\begin{table}[t!]
  \centering
  \caption{Forecast evaluation - weighted autocovariances}
    \begin{tabular}{c|llllllll}
    \multicolumn{9}{c}{}\\
    \multicolumn{9}{c}{Benchmark: RA}\\
    \midrule
   \multicolumn{1}{c|}{\backslashbox{$M$}{Horizon}} & 1      & 2      & 3      & 4      & 5      & 6      & 7      & 8 \\
    \midrule
    $\left\lfloor T^{2/9} \right\rfloor$        & 3.594$^{**}$     & 2.837$^{**}$     & 2.304$^{**}$     & 1.741    & 1.241    & 0.526    & -0.222   & -0.733 \\
     $\left\lfloor T^{1/3} \right\rfloor$         & 3.061$^{**}$     & 2.420$^{**}$     & 1.953$^{*}$     & 1.473    & 1.055    & 0.451    & -0.194   & -0.650 \\
   $\left\lfloor T^{1/2} \right\rfloor$        & 2.430$^{**}$     & 1.917    & 1.551    & 1.182    & 0.857    & 0.368    & -0.159   & -0.551 \\
    T       & 4.976$^{**}$     & 3.772    & 3.098    & 2.378    & 1.712    & 0.665    & -0.251   & -0.870 \\
     \midrule
   \multicolumn{9}{c}{}\\
   \multicolumn{9}{c}{Benchmark: 2\%}\\
     \midrule
     \multicolumn{1}{c|}{\backslashbox{$M$}{Horizon}}  & 1      & 2      & 3      & 4      & 5      & 6      & 7      & 8 \\
    \midrule
    $\left\lfloor T^{2/9} \right\rfloor $       & 4.087$^{**}$     & 3.543$^{**}$     & 3.190$^{**}$     & 2.702$^{**}$     & 2.215$^{**}$     & 1.572    & 0.881    & 0.473 \\
   $\left\lfloor T^{1/3} \right\rfloor$       & 3.526$^{**}$     & 3.059$^{**}$     & 2.734$^{**}$     & 2.312$^{**}$     & 1.898$^{*}$     & 1.355    & 0.770    & 0.418 \\
   $\left\lfloor T^{1/2} \right\rfloor$       & 2.885$^{**}$     & 2.497$^{**}$     & 2.231$^{*}$     & 1.904    & 1.569    & 1.117    & 0.642    & 0.358 \\
    T       & 4.818$^{**}$     & 3.964$^{*}$     & 3.596    & 2.994    & 2.230    & 1.378    & 0.702    & 0.387 \\
    \bottomrule\multicolumn{9}{p{14cm}}{Note: DM test statistic for the null of equal predictive ability of  AR(1) predictions with respect to a rolling average (RA) (top panel) and a constant 2\% (bottom panel) benchmarks. A positive value of the test statistics denotes a larger loss for the benchmark. Long-run variances are computed using the weighted autocovariances in~\eqref{bartlett}. The sample size is 44 and bandwidth values $M$ of $\left\lfloor T^{2/9} \right\rfloor $, $\left\lfloor T^{1/3} \right\rfloor $, $\left\lfloor T^{1/2} \right\rfloor$ and $T$ are respectively 2, 3, 6 and 44.  $^*$ and $^{**}$ denote significance at 10\% and 5\% level using fixed-$b$ critical values.}
    \end{tabular}%
  \label{tab:DMA}%
\end{table}%

\begin{table}[t!]
  \centering
  \caption{Forecast evaluation - weighted periodogram}
    \begin{tabular}{c|llllllll}
    \multicolumn{9}{c}{}\\
    \multicolumn{9}{c}{Benchmark: RA}\\
    \midrule
   \multicolumn{1}{c|}{\backslashbox{$m$}{Horizon}} & 1      & 2      & 3      & 4      & 5      & 6      & 7      & 8 \\
    \midrule
       1 & 1.881 & 1.389 & 1.123 & 0.859 & 0.628 & 0.249 & -0.105 & -0.403 \\
    $\left\lfloor T^{1/4} \right\rfloor $ & 1.736 & 1.397 & 1.141 & 0.896 & 0.671 & 0.300 & -0.135 & -0.500 \\
    $\left\lfloor T^{1/3} \right\rfloor $     & 2.091$^{*}$ & 1.641 & 1.300   & 0.981 & 0.711 & 0.297 & -0.127 & -0.437 \\
    $\left\lfloor T^{1/2} \right\rfloor $    & 2.752$^{**}$ & 2.210$^{**}$ & 1.744 & 1.300   & 0.929 & 0.395 & -0.171 & -0.595 \\
    $\left\lfloor T^{2/3} \right\rfloor $    & 3.665$^{**}$ & 2.941$^{**}$ & 2.360$^{**}$ & 1.783$^{*}$ & 1.269 & 0.530  & -0.220 & -0.725 \\
     \midrule
   \multicolumn{9}{c}{}\\
   \multicolumn{9}{c}{Benchmark: 2\%}\\
     \midrule
     \multicolumn{1}{c|}{\backslashbox{$m$}{Horizon}}  & 1      & 2      & 3      & 4      & 5      & 6      & 7      & 8 \\
    \midrule
    1     & 3.524$^{*}$ & 2.829 & 2.527 & 2.191 & 1.841 & 1.228 & 0.762 & 0.508 \\
    $\left\lfloor T^{1/4} \right\rfloor $     & 2.138$^{*}$ & 1.913 & 1.709 & 1.518 & 1.335 & 1.083 & 0.750  & 0.452 \\
    $\left\lfloor T^{1/3} \right\rfloor $     & 2.564$^{**}$ & 2.238$^{*}$ & 1.925 & 1.592 & 1.294 & 0.916 & 0.546 & 0.303 \\
    $\left\lfloor T^{1/2} \right\rfloor $     & 3.240$^{**}$ & 2.886$^{**}$ & 2.496$^{**}$ & 2.063$^{*}$ & 1.687 & 1.232 & 0.743 & 0.407 \\
    $\left\lfloor T^{2/3} \right\rfloor $    & 4.150$^{**}$ & 3.668$^{**}$ & 3.256$^{**}$ & 2.748$^{**}$ & 2.271$^{**}$ & 1.608 & 0.890  & 0.469 \\
    \bottomrule\multicolumn{9}{p{14cm}}{Note: DM test statistic for the null of equal predictive ability of  AR(1) predictions with respect to a rolling average (RA) (top panel) and a constant 2\% (bottom panel) benchmarks. A positive value of the test statistics denotes a larger loss for the benchmark. Long-run variances are computed using the weighted autocovariance in~\eqref{Daniell}. The sample size is 44 and bandwidth values $m$ of $\left\lfloor T^{1/4} \right\rfloor $, $\left\lfloor T^{1/3} \right\rfloor $, $\left\lfloor T^{1/2} \right\rfloor $ and $\left\lfloor T^{2/3} \right\rfloor$ are respectively 2, 3, 6 and 12. Critical value are obtained from~\eqref{fix_m}.     $^*$ and $^{**}$ denote significance at 10\% and 5\% level using fixed-$m$ critical values.}
    \end{tabular}%
  \label{tab:DMP}%
\end{table}%

We report the outcome of the DM tests  for the null of equal predictive ability of AR(1) predictions with respect to a rolling average and a constant 2\% benchmarks in Tables~\ref{tab:DMA} and\ref{tab:DMP}.

We consider tests in which the DM  statistics are computed estimating the long-run variances as in~\eqref{bartlett} or in \eqref{Daniell}. For $\widehat{\sigma}_A$ we used bandwidths $M$ as $\lfloor T^{2/9} \rfloor$, $\lfloor T^{1/3} \rfloor$, $\lfloor T^{1/2} \rfloor$ and $T$, taking values $2$, $3$, $6$ and $44$ (the case $\lfloor T^{1/4} \rfloor$ is not present as this is again $2$); for $\widehat{\sigma}_P$  the  bandwidth values $m$ of 1, $\left\lfloor T^{1/4} \right\rfloor $, $\left\lfloor T^{1/3} \right\rfloor $, $\left\lfloor T^{1/2} \right\rfloor $ and $\left\lfloor T^{2/3} \right\rfloor$, that for a sample of 44 observations are respectively 2, 3, 6 and 12. In all cases we use critical values from the corresponding fixed smoothing ($\Phi_A(b)$ or $t_{2m}$) distribution.

Results in Tables~\ref{tab:DMA} and \ref{tab:DMP} are in line with our theory, as the outcome of the DM test reflects the autocorrelation documented in  Table~\ref{tab:summary_loss}. In view of the high autocorrelation of $d_t$, the test may be affected by  size distortion, especially with the larger bandwidths $m$, $m=\lfloor T^{1/2} \rfloor$ and $m=\lfloor T^{2/3} \rfloor$, or for short $M$. We therefore discard results for these bandwidths. 
Even results with $m=\lfloor T^{1/3} \rfloor$ or $M=\lfloor T^{1/2} \rfloor$ 
should be considered with caution in this case, especially when the accuracy of the forecasts is compared against the $2 \%$ benchmark, as the autocorrelation seems to be particularly strong in that case. 
Remarkably, this is also the only case in which the equal predictive accuracy null between the AR(1) and the 2$\%$ forecast is significant at $5 \%$ level using $\widehat{\sigma}_P$ and the $m=\lfloor T^{1/3} \rfloor$ bandwidth;
with shorter bandwidths, on the other hand, the null hypothesis of equal predictive accuracy is never rejected at $5\%$ level. This seems to be a disappointing outcome, given the apparent superior performance of forecasts from the AR(1) model at short horizons (as documented in Figure~\ref{fig:Infl} and in Table~\ref{tab:summary_loss}), and we suspect that it is due to the lack of power of the DM test in the presence of autocorrelation. The results when $\widehat{\sigma}_A$ and $M=T$ are used are perhaps slightly more convincing, at least when the one period ahead forecasts are evaluated. Overall, these results highlight how applying the DM test when the loss differential is not weakly dependent may generate unrealiable results.

\section{Conclusion}
In this paper, we have verified that the DM test may be seriously misleading in presence of strong autocorrelation in the loss differential.  \citeasnoun{diebold2015comparing} mentions that  ``[o]f course forecasters may not achieve optimality, resulting in serially correlated, and indeed forecastable, forecast errors. But I(1) nonstationarity of forecast errors takes serial correlation to the extreme''. This is certainly true. However, the DM test is often used against naive benchmarks, for which an I(1) forecast error (or with root close enough to 1, given the sample size) may not be impossible. While this may be seen as an ``abuse'' of the DM test, it seems desirable that a test is robust to such abuse. Our results warn that this is not the case, and that the DM test may perform poorly, generating size distortion or low power, also in the presence weakly dependent processes with autocorrelation close to unity. 

For comparing forecasts, the DM test is ``the only game in town'', as noted in \citeasnoun{diebold2015comparing}. However, one should be aware that the game has its rules. In the empirical application, we used the DM test to compare AR(1) inflation forecasts to two  naive benchmarks. Results indicate that, using a quadratic loss, the test fails exactly because the benchmark forecasts are not optimal under MSE loss, which is not a nice feature. 
This does not mean that one should not use the DM test. Rather, our work suggests that one should take the recommendation in \citeasnoun{diebold2015comparing} to use diagnostic procedures to assess the validity of the assumption of weak dependence of the loss differential very seriously. 

\clearpage
\addcontentsline{toc}{section}{References}
\bibliographystyle{cje}
\bibliography{ref}

@article{Chan_and_Wei,
author = {N. H. Chan and C. Z. Wei},
title = {{Asymptotic Inference for Nearly Nonstationary AR(1) Processes}},
volume = {15},
journal = {The Annals of Statistics},
number = {3},
publisher = {Institute of Mathematical Statistics},
pages = {1050 -- 1063},
year = {1987},
}

@article{patton2007properties,
  title={Properties of optimal forecasts under asymmetric loss and nonlinearity},
  author={Patton, Andrew J and Timmermann, Allan},
  journal={Journal of Econometrics},
  volume={140},
  number={2},
  pages={884--918},
  year={2007},
  publisher={Elsevier}
}

@article{marcellino2003macroeconomic,
  title={Macroeconomic forecasting in the euro area: Country specific versus area-wide information},
  author={Marcellino, Massimiliano and Stock, James H and Watson, Mark W},
  journal={European Economic Review},
  volume={47},
  number={1},
  pages={1--18},
  year={2003},
  publisher={Elsevier}
}

@article{forni2003financial,
  title={Do financial variables help forecasting inflation and real activity in the euro area?},
  author={Forni, Mario and Hallin, Marc and Lippi, Marco and Reichlin, Lucrezia},
  journal={Journal of Monetary Economics},
  volume={50},
  number={6},
  pages={1243--1255},
  year={2003},
  publisher={Elsevier}
}

@Article{clark1999finite,
  author    = {Clark, Todd E},
  title     = {Finite-sample properties of tests for equal forecast accuracy},
  journal   = {Journal of Forecasting},
  year      = {1999},
  volume    = {18},
  number    = {7},
  pages     = {489--504},
  publisher = {Wiley Online Library},
}

@article{coroneo2020comparing,
  title={Comparing predictive accuracy in small samples using fixed-smoothing asymptotics},
  author={Coroneo, Laura and Iacone, Fabrizio},
  journal={Journal of Applied Econometrics},
  volume={35},
  number={4},
  pages={391--409},
  year={2020},
  publisher={Wiley Online Library}
}

@Article{delgado1996optimal,
  author    = {Delgado, Miguel A and Robinson, Peter M},
  title     = {Optimal spectral bandwidth for long memory},
  journal   = {Statistica Sinica},
volume={6},
  pages={97--112},
  year={1996}
}

@Article{diebold2015comparing,
  author    = {Diebold, Francis X},
  title     = {Comparing predictive accuracy, twenty years later: A personal perspective on the use and abuse of {D}iebold--{M}ariano tests},
  journal   = {Journal of Business \& Economic Statistics},
  year      = {2015},
  volume    = {33},
  number    = {1},
  pages     = {1--1},
  publisher = {Taylor \& Francis},
}

@Article{diebold1995comparing,
  author    = {Diebold, Francis X and Mariano, Roberto S},
  title     = {Comparing Predictive Accuracy},
  journal   = {Journal of Business \& Economic Statistics},
  year      = {1995},
volume={20},
  number={1},
  pages     = {253--263},
  publisher = {JSTOR},
}

@Article{giacomini2006tests,
  author    = {Giacomini, Raffaella and White, Halbert},
  title     = {Tests of conditional predictive ability},
  journal   = {Econometrica},
  year      = {2006},
  volume    = {74},
  number    = {6},
  pages     = {1545--1578},
  publisher = {Wiley Online Library},
}

@article{GiraitisPhillips2012Local,
title = {Mean and autocovariance function estimation near the boundary of stationarity},
journal = {Journal of Econometrics},
volume = {169},
number = {2},
pages = {166-178},
year = {2012},
note = {Recent Advances in Nonstationary Time Series: A Festschrift in honor of Peter C.B. Phillips},
issn = {0304-4076},
doi = {https://doi.org/10.1016/j.jeconom.2012.01.020},
url = {https://www.sciencedirect.com/science/article/pii/S0304407612000309},
author = {Liudas Giraitis and Peter C.B. Phillips},
}

@Article{gonccalves2011block,
  author    = {Gon{\c{c}}alves, S{\'\i}lvia and Vogelsang, Timothy J},
  title     = {Block bootstrap HAC robust tests: The sophistication of the naive bootstrap},
  journal   = {Econometric Theory},
  year      = {2011},
  volume    = {27},
  number    = {4},
  pages     = {745--791},
  publisher = {Cambridge University Press},
}

@Article{hualde2017fixed,
  author    = {Hualde, Javier and Iacone, Fabrizio},
  title     = {Fixed bandwidth asymptotics for the studentized mean of fractionally integrated processes},
  journal   = {Economics Letters},
  year      = {2017},
  volume    = {150},
  pages     = {39--43},
  publisher = {Elsevier},
}

@article{Khalaf2017,
title = {Monte Carlo forecast evaluation with persistent data},
journal = {International Journal of Forecasting},
volume = {33},
number = {1},
pages = {1-10},
year = {2017},
issn = {0169-2070},
doi = {https://doi.org/10.1016/j.ijforecast.2016.06.004},
author = {Lynda Khalaf and Charles J. Saunders},
}

@Article{kiefer2005new,
  author    = {Kiefer, Nicholas M and Vogelsang, Timothy J},
  title     = {A new asymptotic theory for heteroskedasticity-autocorrelation robust tests},
  journal   = {Econometric Theory},
  year      = {2005},
  volume    = {21},
  number    = {6},
  pages     = {1130--1164},
  publisher = {Cambridge University Press},
}

@Article{lazarus2018har,
  author    = {Lazarus, Eben and Lewis, Daniel J and Stock, James H and Watson, Mark W},
  title     = {HAR Inference: Recommendations for Practice},
  journal   = {Journal of Business \& Economic Statistics},
  year      = {2018},
  volume    = {36},
  number    = {4},
  pages     = {541--559},
  publisher = {Taylor \& Francis},
}

@article{Moricz2006smoothness,
title = {Absolutely convergent Fourier series and function classes},
journal = {Journal of Mathematical Analysis and Applications},
volume = {324},
number = {2},
pages = {1168-1177},
year = {2006},
author = {Ferenc Móricz},
}

@article{MR2001NarrowBand,
author = {Peter M. Robinson and Domenico Marinucci},
title = {{Narrow-band analysis of nonstationary processes}},
volume = {29},
journal = {The Annals of Statistics},
number = {4},
publisher = {Institute of Mathematical Statistics},
pages = {947 -- 986},
year = {2001},
doi = {10.1214/aos/1013699988},
}

@ARTICLE{McElroyPolitis2012ET,
author={McElroy, Tucker S. and Politis, Dimitris N.},
title={Fixed-b asymptotics for the studentized mean from time series with short, long, or negative memory},
journal={Econometric Theory},
year={2012},
volume={28},
number={2},
pages={471-481},
doi={10.1017/S0266466611000405},
note={cited By 13},
document_type={Article},
source={Scopus},
}

@article{dittmann2002properties,
  title={Properties of nonlinear transformations of fractionally integrated processes},
  author={Dittmann, Ingolf and Granger, Clive WJ},
  journal={Journal of Econometrics},
  volume={110},
  number={2},
  pages={113--133},
  year={2002},
  publisher={Elsevier}
}

@article{kruse2019comparing,
  title={Comparing predictive accuracy under long memory, with an application to volatility forecasting},
  author={Kruse, Robinson and Leschinski, Christian and Will, Michael},
  journal={Journal of Financial Econometrics},
  volume={17},
  number={2},
  pages={180--228},
  year={2019},
  publisher={Oxford University Press}
}

@article{MULLER2007LRVestimates,
title = {A theory of robust long-run variance estimation},
journal = {Journal of Econometrics},
volume = {141},
number = {2},
pages = {1331-1352},
year = {2007},
issn = {0304-4076},
doi = {https://doi.org/10.1016/j.jeconom.2007.01.019},
author    = {M{\"u}ller, Ulrich K.},
}

@Article{muller2014hac,
  author    = {M{\"u}ller, Ulrich K.},
  title     = {H{A}{C} corrections for strongly autocorrelated time series},
  journal   = {Journal of Business \& Economic Statistics},
  year      = {2014},
  volume    = {32},
  number    = {3},
  pages     = {311--322},
  publisher = {Taylor \& Francis},
}

@Article{newey1994automatic,
  author    = {Newey, Whitney K and West, Kenneth D},
  title     = {Automatic lag selection in covariance matrix estimation},
  journal   = {The Review of Economic Studies},
  year      = {1994},
  volume    = {61},
  number    = {4},
  pages     = {631--653},
  publisher = {Wiley-Blackwell},
}

@article{Phillips1987LocUnity,
 ISSN = {00063444},
author = {Peter C. B. Phillips},
 journal = {Biometrika},
 number = {3},
 pages = {535--547},
 publisher = {[Oxford University Press, Biometrika Trust]},
 title = {Towards a Unified Asymptotic Theory for Autoregression},
 volume = {74},
 year = {1987}
}

@article{PhillipsMagd2007MildlyIntegrated,
title = {Limit theory for moderate deviations from a unit root},
journal = {Journal of Econometrics},
volume = {136},
number = {1},
pages = {115-130},
year = {2007},
issn = {0304-4076},
doi = {https://doi.org/10.1016/j.jeconom.2005.08.002},
author = {Peter C.B. Phillips and Tassos Magdalinos},
}

@Article{PM2007Chapterinbook,
  author  = {Phillips, Peter C.B. and Magdalinos, Tassos},
  title   = {Limit theory for moderate deviations from
a unit root under weak dependence},
  journal = {In: Phillips, G.D.A., Tzavalis, E. (Eds.), The Refinement of Econometric Estimation and Test Procedures: Finite Sample and
Asymptotic Analysis. Cambridge University Press, Cambridge},
  year    = {2007},
}

@article{PhillipsShimotsu2004LW,
author = {Peter C. B. Phillips and Katsumi Shimotsu},
title = {{Local Whittle estimation in nonstationary and unit root cases}},
volume = {32},
journal = {The Annals of Statistics},
number = {2},
publisher = {Institute of Mathematical Statistics},
pages = {656 -- 692},
  year      = {2004},
}

@article{PhillipsSolo1992,
author = {Peter C. B. Phillips and Victor Solo},
title = {{Asymptotics for Linear Processes}},
volume = {20},
journal = {The Annals of Statistics},
number = {2},
publisher = {Institute of Mathematical Statistics},
pages = {971 -- 1001},
year = {1992},
}

@article{Robinson95log-p,
author = {Peter M. Robinson},
title = {{Log-Periodogram Regression of Time Series with Long Range Dependence}},
volume = {23},
journal = {The Annals of Statistics},
number = {3},
publisher = {Institute of Mathematical Statistics},
pages = {1048 -- 1072},
year = {1995},
doi = {10.1214/aos/1176324636},
URL = {https://doi.org/10.1214/aos/1176324636}
}

@article{Robinson95LW,
author = {Peter M. Robinson},
title = {{Gaussian Semiparametric Estimation of Long Range Dependence}},
volume = {23},
journal = {The Annals of Statistics},
number = {5},
publisher = {Institute of Mathematical Statistics},
pages = {1630 -- 1661},
year = {1995},
doi = {10.1214/aos/1176324317},
}

@article{rossi2005testing,
  title={Testing long-horizon predictive ability with high persistence, and the Meese--Rogoff puzzle},
  author={Rossi, Barbara},
  journal={International Economic Review},
  volume={46},
  number={1},
  pages={61--92},
  year={2005},
  publisher={Wiley Online Library}
}

@article{Shao2015Selfnormalization,
author = {Xiaofeng Shao},
title = {Self-Normalization for Time Series: A Review of Recent Developments},
journal = {Journal of the American Statistical Association},
volume = {110},
number = {512},
pages = {1797-1817},
year  = {2015},
publisher = {Taylor & Francis},
doi = {10.1080/01621459.2015.1050493},
}

@Article{sun2014let,
  author    = {Sun, Yixiao},
  title     = {Let’s fix it: Fixed-b asymptotics versus small-b asymptotics in heteroskedasticity and autocorrelation robust inference},
  journal   = {Journal of Econometrics},
  year      = {2014},
  volume    = {178},
  pages     = {659--677},
  publisher = {Elsevier},
}

@Article{sun2014fixed,
  author    = {Sun, Yixiao},
  title     = {Fixed-smoothing asymptotics in a two-step generalized method of moments framework},
  journal   = {Econometrica},
  year      = {2014},
  volume    = {82},
  number    = {6},
  pages     = {2327--2370},
  publisher = {Wiley Online Library},
}

@INCOLLECTION{Sun2014EssayPCB,
title = {Fixed-smoothing Asymptotics and Asymptotic: F: and: t: Tests in the Presence of Strong Autocorrelation},
author = {Sun, Yixiao},
year = {2014},
pages = {23-63},
booktitle = {Essays in Honor of Peter C. B. Phillips},
volume = {33},
publisher = {Emerald Publishing Ltd},
}

@article{corradi2001predictive,
  title={Predictive ability with cointegrated variables},
  author={Corradi, Valentina and Swanson, Norman R and Olivetti, Claudia},
  journal={Journal of Econometrics},
  volume={104},
  number={2},
  pages={315--358},
  year={2001},
  publisher={Elsevier}
}

@article{robinson1983review,
  title={Review of various approaches to power spectrum estimation},
  author={Robinson, PM},
  journal={Handbook of Statistics},
  volume={3},
  pages={343--368},
  year={1983},
  publisher={Elsevier}
}

@article{mccracken2020diverging,
  title={Diverging tests of equal predictive ability},
  author={McCracken, Michael W},
  journal={Econometrica},
  volume={88},
  number={4},
  pages={1753--1754},
  year={2020},
  publisher={Wiley Online Library}
}

@article{choe2023comparing,
  title={Comparing sequential forecasters},
  author={Choe, Yo Joong and Ramdas, Aaditya},
  journal={Operations Research},
  year={2023},
  publisher={INFORMS}
}

@article{henzi2021valid,
    author = {Henzi, Alexander and Ziegel, Johanna F},
    title = "{Valid sequential inference on probability forecast performance}",
    journal = {Biometrika},
    volume = {109},
    number = {3},
    pages = {647-663},
    year = {2021},
    month = {09},
    issn = {1464-3510},
    doi = {10.1093/biomet/asab047},
    url = {https://doi.org/10.1093/biomet/asab047},
    eprint = {https://academic.oup.com/biomet/article-pdf/109/3/647/45512191/asab047.pdf},
}

\clearpage
\appendix
\section{Derivations}\label{app_proofs}
We provide here a more detailed derivation of some of the results that we claimed in the paper, with accompanying regularity conditions when needed. When we establish bounds we occasionally use $C<\infty$ as a finite bound, not necessarily the same one in every case. Recall that $I\left(\lambda_j\right)$ is the periodogram if $d_t$.

We discuss the results for $DM_A$ and $DM_P$ separately, starting from $DM_P$.

\subsection{Results for Subsection \ref{Subsec_Local}}\label{app_ltu}
All the results in this subsection are for $d_t$ generated as in \eqref{dt}-\eqref{rho_T},  with $c \leq 0$, under Assumption~\ref{ass1}.

\begin{lemma}
\label{ExpectedPeriodogram}
\begin{equation*} 
    I\left(\lambda_j\right)=O_p \left(j^{-2} T^2 \right)
\end{equation*}
\end{lemma}

\begin{proof} 
We first present the proof for $c<0$. \\
Denoting $f\left(\lambda\right)$ as the spectral density of $y_t$, $g\left(\lambda\right)$ as the density of $u_t$, and \begin{equation*}
f^\star(\lambda)=\left| 1- \rho \, e ^ {-i \lambda} \right|^{-2}
=\frac{1}{\upsilon^2+2 \rho (1-cos(\lambda))},
\end{equation*} 
where $    \upsilon=1-\rho$, 
then  $f(\lambda)=f^\star(\lambda)g(\lambda)$.

We use for $f^\star(\lambda)$ the same bound as in  \citeasnoun{GiraitisPhillips2012Local}: for $|\rho|<1$, $\lambda \leq \pi$ 
\begin{equation}
    f^\star(\lambda) \leq \frac{1}
    { \upsilon^2 + \rho \lambda^2 /3 }. \label{GirBound}
\end{equation}
Notice that we dropped the reference to $T$ in $\rho_T$ to simplify the notation and to align it to \citeasnoun{GiraitisPhillips2012Local}.

We follow closely the proof in \citeasnoun{Robinson95log-p}, but our proof is easier as we only need to establish an upper bound. 
Then
\begin{equation*}
    E\left(I\left(\lambda_j\right)\right)=\int_{- \pi} ^{\pi} f(\lambda) K (\lambda-\lambda_j)d \lambda
\end{equation*}
where $K(\lambda)$ is proportional to the Fej\'{e}r's kernel, $K(\lambda)=(2 \pi T)^{-1}  \left|\sum\sum_{t,s=1}^T e^{i (t-s)\lambda}\right|^2$. \\
Proceeding as in \citeasnoun{Robinson95log-p} we then partition the integral as 
\begin{equation*}
    \int_{- \pi}^{\pi} = 
    \int_{- \pi}^{- \lambda_j /2} 
    + \int_{- \lambda_j /2}^{\lambda_j /2} 
    + \int_{\lambda_j /2}^{2 \lambda_j} 
    +\int_{2 \lambda_j }^{\pi} 
\end{equation*}
and discuss them separately.

\begin{equation*}
         \int_{- \pi}^{- \lambda_j /2} f(\lambda) K (\lambda-\lambda_j)d \lambda  
    \leq C \ \{sup_{\lambda \in [\lambda_j/2,\pi]} f^\star(\lambda) \} 
    \int_{ \lambda_j /2}^{ \pi}  K (\lambda+\lambda_j)d \lambda 
    \leq C \ \lambda_j^{-2} \ j^{-1}=O(j^{-3}T^2) 
\end{equation*}
where we used the bounds $g(\lambda) \leq C$, $sup_{\lambda \in [\lambda_j/2,\pi]} f^\star(\lambda)  \leq C \lambda_j^{-2}$ from \eqref{GirBound} and $\int_{ \lambda_j /2}^{ \pi}  K (\lambda+\lambda_j)d \lambda=O(j^{-1})$ as in \citeasnoun{Robinson95log-p}, page 1061 in the text above (4.6); the bound $\int_{2 \lambda_j }^{\pi}=O(j^{-3}T^2)$ can be established in the same way. Next, 
\begin{equation*}
\int_{- \lambda_j /2}^{\lambda_j /2} f(\lambda) K (\lambda-\lambda_j)d \lambda  
\leq  \int_{- \lambda_j /2}^{\lambda_j /2} f(\lambda) d \lambda \  \left\{sup_{\lambda \in [-\lambda_j/2, \lambda_j/2]} K (\lambda-\lambda_j) \right\} 
=O (T \times T^{-1} \lambda_j^{-2})=O(T^2 j^{-2})
\end{equation*}
where we bounded $\{ sup_{\lambda \in [-\lambda_j/2, \lambda_j/2]} K (\lambda-\lambda_j) \} 
=O ( T^{-1} \lambda_j^{-2})$ as in \citeasnoun{Robinson95log-p} and $\int_{- \lambda_j /2}^{\lambda_j /2} f(\lambda) d \lambda \leq Var(y_t)=O(T)$. Finally, 
\begin{equation*}
\int_{\lambda_j /2}^{ 2\lambda_j} f(\lambda) K (\lambda-\lambda_j)d \lambda  
\leq C \left\{ \sup_{\lambda \in [\lambda_j/2, 2 \lambda_j]} f^\star(\lambda) \right\}  \int_{ \lambda_j /2}^{2 \lambda_j }   \  K (\lambda-\lambda_j) d \lambda\ 
=O (\lambda_j^{-2}) 
\end{equation*}
where we bounded $\int_{ \lambda_j /2}^{2 \lambda_j }   \  K (\lambda-\lambda_j) d \lambda = O(1) $. This completes the proof for $c<0$.

When $c=0$ we rewrite, as in Lemma A.1 of \citeasnoun{PhillipsShimotsu2004LW} for the unit root case, 
\begin{equation*}
(1-e^{i \lambda}) w(\lambda)
=w_u(\lambda)-\frac{e^{i \lambda}}{\sqrt{2 \pi T}} (e^{i T \lambda} y_T - y_0)
\end{equation*}
where $w(\lambda)$ and $w_u(\lambda)$ is the Fourier transform of $y_t$ and $u_t$, respectively. Thus, bounding 
$E|w_u(\lambda)|=O((E(|w_u(\lambda)|^2))^{1/2})=O(1)$, $E|y_T|=O((E(y_T^2))^{1/2})=O(T^{1/2})$, $|(1-e^{i \lambda})|^{-2} < C \lambda^{-2}$, the result follows immediately.
\end{proof}
\smallskip

\begin{lemma} \label{AveragePeriodogram}
\begin{equation*}
\text{  As } m \rightarrow \infty, \  m/T \rightarrow 0, \ \  \frac{1}{T^2} 2 \pi \sum_{j=1}^m {I(\lambda_j}) \Rightarrow \omega^2    {\frac{1}{2}\int_0 ^1  (J_c(r)-\overline{J_c})^2 dr}  
\end{equation*}
\end{lemma}

\begin{proof}
We rewrite  
\begin{equation*}
    \frac{1}{T^2} 2 \pi \sum_{j=1}^m {I(\lambda_j})=
    \frac{1}{T^2} 2 \pi \sum_{j=1}^{T/2} {I(\lambda_j})
    -\frac{1}{T^2} 2 \pi \sum_{j=m+1}^{T/2} {I(\lambda_j)}
\end{equation*}
and notice that 
\begin{equation*}
    \frac{1}{T^2} 2 \pi \sum_{j=1}^{T/2} {I(\lambda_j})
    =\frac{1}{T^2} \frac{1}{2} \sum_{t=1}^T {(y_t-\overline{y})^2} \Rightarrow {\omega^2 \frac{1}{2}\int_0 ^1  (J_c(r)-\overline{J_c})^2 dr}
\end{equation*}
using \eqref{Lemma1bP} and \eqref{Lemma1cP}, while
\begin{equation*}
    \frac{1}{T^2} 2 \pi \sum_{j=m+1}^{T/2} {I(\lambda_j)}=O_p(T^{-2} m^{-1} T^2)=O_p( m^{-1})=o_p(1)
\end{equation*}
using Lemma \ref{ExpectedPeriodogram} (notice that that result is not restricted to a band of frequencies degenerating to 0). \\ The result when $c=0$ can be deduced from \citeasnoun{MR2001NarrowBand}: their moments condition is stronger and their proof is more complex than the argument given here because they established more results. 
\end{proof}

\textbf{Proof of Theorem}~\ref{DF-standard} - \textit{Case P} \begin{proof}

\begin{equation*}
\frac{1}{\sqrt{m}} DM_P
= \frac{1}{\sqrt{m}} \sqrt{T} \frac{\overline{y}+\mu}
{\sqrt{\frac{1}{m} 2 \pi \sum_{j=1}^m {I(\lambda_j)} }}
=\frac{\frac{\sqrt{T}}{T}(\overline{y}+\mu)}
{{\sqrt{m}} \sqrt{\frac{1}{T^2}\frac{1}{m} 2 \pi \sum_{j=1}^m {I(\lambda_j)} }}\end{equation*}
where in particular notice that $\frac{\sqrt{T}}{T}\mu \rightarrow 0$
so $ \frac{\sqrt{T}}{T}(\overline{y}+\mu) \Rightarrow \omega \int_0^1 J_c(r)dr$ by a standard FCLT even when $\mu \neq 0$. 
The result thus follows from the convergence in \eqref{Lemma1bP}, Lemma \ref{AveragePeriodogram} and the continuous mapping theorem.
\end{proof}

\textbf{Proof of Theorem}~\ref{DF-fixed}  - \textit{Case P}  
\begin{proof} 
The proof proceeds in the same way as in Theorem \ref{DF-standard}, but in this case we consider $m$ as finite. 
\begin{equation*}
DM_P
= \sqrt{T} \frac{\overline{y}+\mu}
{\sqrt{\frac{1}{m} 2 \pi \sum_{j=1}^m {I(\lambda_j)} }}
=\frac{\frac{\sqrt{T}}{T}(\overline{y}+\mu)}
{\sqrt{\frac{1}{T^2}\frac{1}{m} 2 \pi \sum_{j=1}^m {I(\lambda_j)} }}\end{equation*}
where the numerator is discussed as in as in Theorem \ref{DF-standard}, and ${\frac{1}{T^2}\frac{1}{m} 2 \pi \sum_{j=1}^m {I(\lambda_j)} } \Rightarrow \omega^2 \frac{1}{m} 2 \pi \sum_{j=1}^m Q_{P}(j)$ using the same argument as in Lemma 1 of \citeasnoun{hualde2017fixed}.
\end{proof}

\begin{lemma}\label{autocovs}
\begin{equation*}
\text{As } M\rightarrow \infty, \ M/T \rightarrow 0,  \ \  \frac{1}{T M} \widehat{\sigma}^2_A \rightarrow_d \omega^2 \int_0^1 \left( {J}_c(r)^2-\overline{J}_c^2 \right) dr
\end{equation*}
\end{lemma}
\begin{proof}
    Notice that
\begin{equation*}
\widehat{\gamma}_l=\frac{1}{T} \sum_{t=l+1}^T 
(d_t-\overline{d})(d_{t-l}-\overline{d})=\frac{1}{T} \sum_{t=l+1}^T 
(y_t-\overline{y})(y_{t-l}-\overline{y})
\end{equation*}
Using recursive substitutions, 
\begin{equation*}
y_t=\rho^l y_{t-l} + \sum_{j=0}^{l-1}\rho^j u_{t-j}=y_{t-l}-(1-\rho^l) y_{t-l} + \sum_{j=0}^{l-1}\rho^j u_{t-j}
\end{equation*}
and
\begin{equation*}
y_t-\overline{y}=(y_{t-l}-\overline{y})-(1-\rho^l) y_{t-l} + \sum_{j=0}^{l-1}\rho^j u_{t-j}
\end{equation*}
so that 
\begin{equation}\label{gamma_l}
\widehat{\gamma}_l=\frac{1}{T} \sum_{t=l+1}^T 
(y_{t-l}-\overline{y})^2-(1-\rho^l) \frac{1}{T} \sum_{t=l+1}^T  y_{t-l}(y_{t-l}-\overline{y})+\frac{1}{T} \sum_{t=l+1}^T\left(\sum_{j=0}^{l-1}\rho^j u_{t-j}\right)(y_{t-l}-\overline{y}).
\end{equation}
We discuss the three terms in \eqref{gamma_l} separately. For the first one, 
\begin{equation*}
\frac{1}{T} \sum_{t=l+1}^T (y_{t-l}-\overline{y})^2 
=\frac{1}{T} \sum_{t=1}^T (y_{t}-\overline{y})^2 - \frac{1}{T} \sum_{t=1}^l (y_{t}-\overline{y})^2
\end{equation*}
and notice that 
\begin{equation*}
\frac{1}{T^2}  \sum_{t=1}^T (y_{t}-\overline{y})^2 \rightarrow_d \omega^2 \int_0^1 \left( {J}_c(r)^2-\overline{J}_c^2 \right) dr
\end{equation*}
and, in view of Lemma 1 of \citeasnoun{Phillips1987LocUnity}
\begin{equation*}
 \frac{1}{T^2} \sum_{t=1}^l (y_{t}-\overline{y})^2=O_p\left(\frac{1}{T^2} l (T^{1/2})^2\right)=O_p\left(\frac{M}{T}\right).
\end{equation*}
For the second term in \eqref{gamma_l}, using the mean value theorem expansion 
\begin{equation*}
1-\rho^l=1-e^{l c/T}=l c_m/T \text{ for } c \leq\ c_m\leq 0 
\end{equation*}
and bounds from Lemma 1 in \citeasnoun{Phillips1987LocUnity} then 
\begin{equation*}
(1-\rho^l) \frac{1}{T^2} \sum_{t=l+1}^T  y_{t-l}(y_{t-l}-\overline{y})=O_p\left(\frac{M}{T}\right)
\end{equation*}
Finally, as $E\left( (\sum_{j=0}^{l-1}\rho^j u_{t-j})^2 \right)=O(l)$, then with an application of the Cauchy-Schwarz inequality, 
\begin{equation*}
\frac{1}{T^2} \sum_{t=l+1}^T\left(\sum_{j=0}^{l-1}\rho^j u_{t-j}\right)(y_{t-l}-\overline{y}).=O_p\left(\frac{M^{1/2}}{T^{1/2}}\right)
\end{equation*}
so 
\begin{equation*}
\frac{1}{T} \widehat{\gamma}_l \rightarrow_d \omega^2 \int_0^1 \left( {J}_c(r)^2-\overline{J}_c^2 \right) dr
\end{equation*}
Thus, recalling $\sum_{l=1}^M \frac{M-l}{M}=1/2 \frac{M(M+1)}{M}=1/2(M+1)$,
\begin{equation*}
\frac{1}{T M} \widehat{\sigma}_A^2 \rightarrow_d \omega^2 \int_0^1 \left( {J}_c(r)^2-\overline{J}_c^2 \right) dr
\end{equation*}
\end{proof}

\textbf{Proof of Theorem}~\ref{DF-standard} - \textit{Case A} \begin{proof}
The proof proceeds in the same way as for Case A, but using the limit in Lemma 
 \ref{autocovs} instead.
\end{proof}

\textbf{Proof of Theorem}~\ref{DF-fixed}  - \textit{Case A}  
\begin{proof} 
Proceeding as in \citeasnoun{kiefer2005new} 
\begin{equation*}
\frac{1}{T}\widehat{\sigma}_A^2 \rightarrow_d \omega^2 Q_A(b)^2
\end{equation*}
The proof then proceeds in the same way as for Case A.
\end{proof}

\subsection{Results for Subsection \ref{Subsec_mildly}}\label{app_mod}
All the results in this subsection are for $d_t$ generated as in \eqref{dt}, \eqref{ar1}, \eqref{rho_mild}  with $c < 0$, under Assumption~\ref{ass2}. 

We first state some properties that can be derived from Assumption~\ref{ass2} 
\begin{lemma}\label{Regarding_Ass2} \text{ }
\begin{itemize} 
    \item[i.]  The condition in Assumption~\ref{ass2} implies the condition in~\ref{ass1}.
    \item[ii.] The condition in Assumption~\ref{ass2}  is stronger than Assumption~\ref{ass1}.
    \item[iii.] $|g(x+h)-g(x)|<C |h|$ for all $x$, $h$.

\end{itemize}
\end{lemma}

\begin{proof}
   The condition in Assumption~\ref{ass2} implies the condition in~\ref{ass1}. By summation by parts, for any $n$,
\begin{align*}
& \sum_{j=0}^{n}
 j^{1/2}|\psi_j| 
 =\sum_{j=1}^{n} j^{1/2}|\psi_j|
 =1\sum_{j=1}^{n}|\psi_j|+(\sum_{j=1}^{n-1}{((j+1)^{1/2}-j^{1/2})}  \sum_{s=j+1}^{n} |\psi_s| \\
 & \sum_{j=0}^{n}
 j^{1/2}|\psi_j| \leq 
\sum_{j=1}^{n}|\psi_j|+C\sum_{j=1}^{n-1}{j^{-1/2}}  \sum_{s=j+1}^{n} |\psi_s|
\leq 
\sum_{j=1}^{\infty}|\psi_j|+C\sum_{j=1}^{\infty}{j^{-1/2}}  \sum_{s=j+1}^{\infty} |\psi_s| 
\end{align*}
so 
\begin{equation*}
 \sum_{j=1}^{\infty}j^{1/2}|\psi_j|     \leq 
C+C\sum_{j=1}^{\infty}{j^{-1/2}}  j^{-1-\alpha}<C
\end{equation*}
To see that the reverse is not true, which means that Assumption~\ref{ass2}  really strengthens~\ref{ass1}, notice that $\psi_j=(j+1)^{-3/2-\eta}$ for $\eta>0$ and suitably small meets Assumption~\ref{ass1} but not~\ref{ass2}.

Moreover,
another application of summation by parts gives 
\begin{align*}
  &  \sum_{k=1}^n  {k |\psi_{j+k}|} = 1 \times \sum_{k=1}^n  { |\psi_{j+k}|} + 
    \{\sum_{k=1}^{n-1} {((k+1)-k)} \sum_{s=k+1}^{n} {|\psi_{j+s}|} \} =  \sum_{k=1}^n  { |\psi_{j+k}|} + 
    \sum_{k=1}^{n-1}  \sum_{s=k+1}^{n} {|\psi_{j+s}|} \\
    &
    \sum_{k=1}^\infty  {k |\psi_{j+k}|} 
    \leq  \sum_{k=1}^\infty  { |\psi_{j+k}|} + 
    \sum_{k=1}^{\infty}  \sum_{s=k+1}^{\infty} {|\psi_{j+k}|}
    \leq C
    j^{-1-a} +
    \sum_{k=1}^{\infty} 
    (k+j)^{-1-a}
    \leq C
    j^{-1-a} + C j^{-a}
\end{align*}
so Assumption~\ref{ass2} also implies that 
\begin{equation*}
    \sum_{k=0}^\infty {k |\psi_k|} < \infty
\end{equation*}
which is a sufficient condition for Theorem 2.1 of \citeasnoun{Moricz2006smoothness}, see page 1169. 
Denoting $g(\lambda)$ the spectral density of $u_t$, Theorem 2.1 of   implies \citeasnoun{Moricz2006smoothness} that $g(\lambda)$ belongs either to the Lipschitz class Lip(1), for which 
\begin{equation*}
    |g(x+h)-g(x)|<C |h|
\end{equation*}
for all $x$, $h$.
\end{proof}

\begin{lemma} \label{LimitVariance}
\textit{As $T \rightarrow  \infty$}
\begin{equation}
\frac{1}{T^{1+\alpha}} \sum_{t=1}^T {(d_t-\overline{d})^2}=\frac{1}{T^{1+\alpha}} \sum_{t=1}^T {(d_t-\mu)^2}+o_p(1) \rightarrow_p \frac{\omega^2}{-2c}     
\end{equation}
\end{lemma}

\begin{proof} The limit
\begin{equation*}
    \frac{1}{T^{1+\alpha}} \sum_{t=1}^T {(d_t-\mu)^2} \rightarrow_p \frac{\omega^2}{-2c}
\end{equation*}
is already in \citeasnoun{PM2007Chapterinbook}; it can also be derived from (2.13) and (2.16) of \citeasnoun{GiraitisPhillips2012Local}, setting $\upsilon=1-\rho_T=cT^{-\alpha}$
and taking the limit for $T \rightarrow \infty$.\\
Next, rewriting $\sum_{t=1}^T {(d_t-\overline{d})^2}= \sum_{t=1}^T {(d_t-\mu)^2}-T{(\overline{d}-\mu)^2}$,
from the CLT on Theorem 2.1 of \citeasnoun{GiraitisPhillips2012Local}, ${(\overline{d}-\mu)^2}=O_p\left(T^{2\alpha-1}\right)$, so 
\begin{equation*}
  \frac{1}{T^{1+\alpha}} \sum_{t=1}^T {(d_t-\overline{d})^2}-\frac{1}{T^{1+\alpha}} \sum_{t=1}^T {(d_t-\mu)^2}=O_p(T^{-1-\alpha} \ T \ T^{2 \alpha -1} )=O_p(T^{ \alpha -1} ).
\end{equation*}
  The lemma is established as we combine these results.
\end{proof}
  
  \begin{lemma} \label{BoundPeriodogramMildly}
\begin{equation*} 
    I\left(\lambda_j\right)=
    O_p \left(j^{-2} T^2 \right)
\end{equation*} 
\end{lemma}
\begin{proof} The proof follows as in Lemma \ref{ExpectedPeriodogram}, but using the bound $\int_{-\pi}^{\pi}f(\lambda)d\lambda=O(T^\alpha)$. 
\end{proof}

\begin{lemma} \label{Periodogram_Mildly}
For frequencies $\lambda_j$ such that $j<m$, $m T^{\alpha-1} \rightarrow 0$ as $T \rightarrow \infty$,
\begin{equation*} 
     f(\lambda_j)^{-1} I\left(\lambda_j\right)
    =1+O_p (j^{-1} \ln(j+1))
\end{equation*} 
\end{lemma}

\begin{proof}
We consider 
\begin{equation*}
    f(\lambda_j)^{-1}    E\left(I\left(\lambda_j\right)\right)-1
    =f(\lambda_j)^{-1}
    \int_{- \pi} ^{\pi} \left(f(\lambda)-f(\lambda_j)\right) K (\lambda-\lambda_j)d \lambda
\end{equation*}
and again we evaluate this integral over subsets of $(-\pi, \pi)$ as in Lemma \ref{ExpectedPeriodogram}. 
Then, 
\begin{equation}\label{BoundL4case1}
    f(\lambda_j)^{-1}
    \int_{- \pi} ^{-\lambda_j/2} \left(f(\lambda)-f(\lambda_j)\right) K (\lambda-\lambda_j)d \lambda
    \leq 
    C \ f^\star(\lambda_j)^{-1} \{sup_{\lambda \in [\lambda_j/2,\pi]} f^\star(\lambda) \} 
    \int_{ \lambda_j /2}^{ \pi}  K (\lambda+\lambda_j)d \lambda 
\end{equation}
where again we used $ \{sup_{\lambda \in [\lambda_j/2,\pi]} f(\lambda) \} \leq \{sup_{\lambda \in [\lambda_j/2,\pi]} g(\lambda) \} \{sup_{\lambda \in [\lambda_j/2,\pi]} f^\star(\lambda) \}$
and $g(\lambda)<C$ uniformly in $\lambda$ and   $g(\lambda_j)^{-1}<C$. Recalling that $f^\star(\lambda_j)^{-1}={\upsilon^2+2 \rho (1-cos(\lambda_j))}$, we bound $ f^\star(\lambda_j)^{-1} \leq  \upsilon^2+ C (sin(\lambda_j/2))^2 \leq \upsilon^2+ C (\lambda_j/2)^2 \leq \upsilon^2+ C \lambda_j^2 $ where we used $sin(\lambda_j/2) \sim \lambda_j/2 $ as $j/T \rightarrow 0 $. Then, 
\begin{align*}
& \ f^\star(\lambda_j)^{-1} \{sup_{\lambda \in [\lambda_j/2,\pi]} f^\star(\lambda) \} 
\leq
(\upsilon^2+ C \lambda_j^2)
\{sup_{\lambda \in [\lambda_j/2,\pi]} f^\star(\lambda) \} \\
& \leq 
\upsilon^2 \{sup_{\lambda \in [\lambda_j/2,\pi]} f^\star(\lambda) \} 
+ C \lambda_j^2 
\{sup_{\lambda \in [\lambda_j/2,\pi]} f^\star(\lambda) \} 
\leq 
\upsilon^2 \upsilon^{-2} 
+ C \lambda_j^2 
\lambda_j^{-2} \leq C
\end{align*}
using \eqref{GirBound}.
Thus, recalling $\int_{ \lambda_j /2}^{ \pi}  K (\lambda+\lambda_j)d \lambda =O(j^{-1})$, then $\eqref{BoundL4case1}=O(j^{-1})$. The bound $ f(\lambda_j)^{-1}\int_{2 \lambda_j }^{\pi}=O(j^{-1})$ can be established in the same way. 
\\
Next, 
\begin{equation}\label{BoundL4case2}
f(\lambda_j)^{-1} \int_{- \lambda_j /2}^{\lambda_j /2} \left(f(\lambda)-f(\lambda_j)\right) K (\lambda-\lambda_j)d \lambda  
\leq C \ f{^\star(\lambda_j)}^{-1}  \{sup_{\lambda \in [-\lambda_j/2, \lambda_j/2]} 
K(\lambda-\lambda_j) \} 
\int_{- \lambda_j /2}^{\lambda_j /2} \left(f(\lambda)-f(\lambda_j)\right)d \lambda  
\end{equation}
where again we bound $sup_{\lambda \in [-\lambda_j/2, \lambda_j/2]} K (\lambda-\lambda_j) \} 
=O ( T^{-1} \lambda_j^{-2})$; moreover, 
\begin{equation*}
    \int_{- \lambda_j /2}^{\lambda_j /2} \left(f(\lambda)-f(\lambda_j)\right)d \lambda = \int_{- \lambda_j /2}^{\lambda_j /2} \left(f(\lambda) -f(0)+f(0)-f(\lambda_j)\right)d \lambda 
    \leq 
    2 \int_{0}^{\lambda_j /2} \left|f(\lambda) -f(0)\right| d\lambda+\left|f(0)-f(\lambda_j)\right|\lambda_j
\end{equation*}
From \citeasnoun{GiraitisPhillips2012Local}, page 175, 
\begin{equation*}
    \left|f(\lambda) -f(0)\right|
    \leq \lambda^2 f^\star(\lambda) \upsilon^{-2}+ \lambda^2  \upsilon^{-2}
\end{equation*}
so 
\begin{equation*}
   2 \int_{0}^{\lambda_j /2} \left|f(\lambda) -f(0)\right| d\lambda+\left|f(0)-f(\lambda_j)\right|\lambda_j 
    \leq     C  \lambda_j  \upsilon^{-2}
\end{equation*}
and \eqref{BoundL4case2} is bounded by 
\begin{equation*}
    C \ (\upsilon^2+\lambda_j^2)
  \  T^{-1} \lambda_j^{-2}
 \  \  \lambda_j  \upsilon^{-2} 
    \leq C (T^{-1} \lambda_j^{-1}
+ T^{-1} \lambda_j  \upsilon^{-2}) =O(j^{-1}+j^{-1}(jT^{\alpha-1})^2)
\end{equation*}
where the last bound is $o(j^{-1})$ because $j \leq m$ and $mT^{\alpha-1}\rightarrow0$.
\\
For the next integral we introduce $f^\star(\lambda)'=\frac{\partial f^\star(\lambda)}{\partial \lambda}$, where  
$f^\star(\lambda)'=-(f^\star(\lambda))^2 2\rho \sin(\lambda)$ and
rewrite 
\begin{equation*}
    f(\lambda)-f(\lambda_j)
    =f^\star(\lambda)g(\lambda)
    -f^\star(\lambda_j)g(\lambda)
    +f^\star(\lambda_j)g(\lambda)
    -f^\star(\lambda_j)g(\lambda_j)
\end{equation*}
then 
\begin{align}
& f(\lambda_j)^{-1} \int_{\lambda_j /2}^{ 2\lambda_j} \left(f(\lambda)-f(\lambda_j)\right) K (\lambda-\lambda_j)d \lambda  \nonumber
\\
& \leq C f^\star(\lambda_j)^{-1} \left\{ \sup_{\lambda \in [\lambda_j/2, 2 \lambda_j]} \left|f^\star(\lambda)'\right| \right\}  \int_{ \lambda_j /2}^{2 \lambda_j }  |\lambda - \lambda_j|
\  K (\lambda-\lambda_j) d \lambda\ \label{BoundL4case3}
 \\ 
 & + C
 f^\star(\lambda_j)^{-1}  f^\star(\lambda_j)   
 \int_{ \lambda_j /2}^{2 \lambda_j }  |\lambda - \lambda_j|
\  K (\lambda-\lambda_j) d \lambda\ \label{BoundL4case4}
\end{align}
where notice that 
\begin{equation*}
\left\{ \sup_{\lambda \in [\lambda_j/2, 2 \lambda_j]} \left|f^\star(\lambda)'\right| \right\} \leq C f^\star(\lambda_j)^2 \lambda_j. 
\end{equation*}
 
The bound in \eqref{BoundL4case3} is 
\begin{align*}
& C f^\star(\lambda_j)  \lambda_j  \int_{ \lambda_j /2}^{2 \lambda_j }  |\lambda - \lambda_j|
\  K (\lambda-\lambda_j) d \lambda\   
\leq
C \lambda_j^{-1}  \int_{ \lambda_j /2}^{2 \lambda_j }  |\lambda - \lambda_j|
\  K (\lambda-\lambda_j) d \lambda\   \\
& =O \left((j/T)^{-1}T^{-1} \ln(j+1) \right)=O(j^{-1} \ln(j+1))
\end{align*}
    where we used
\begin{equation*}
\int_{ \lambda_j /2}^{2 \lambda_j }  |\lambda - \lambda_j|
\  K (\lambda-\lambda_j) d \lambda=O \left(T^{-1} \ln(j+1) \right)   
\end{equation*}
 as in \citeasnoun{Robinson95log-p}; proceeding in the same way, 
 the bound in \eqref{BoundL4case4} is 
\begin{equation*}
C   \int_{ \lambda_j /2}^{2 \lambda_j }  |\lambda - \lambda_j|
\  K (\lambda-\lambda_j) d \lambda\   
=O(T^{-1} \ln(j+1)).
\end{equation*}
\end{proof}

\begin{lemma} \label{Daniell_Mildly}
\textit{For $m \rightarrow \infty$, $m \ T^{\alpha-1} \rightarrow 0$ as $T \rightarrow \infty$,
\begin{equation*} 
    \frac{1}{m} \sum_{j=1}^{m} 
    ( f\left(\lambda_j\right)^{-1}
    I\left(\lambda_j\right) - 1 ) 
    =o_p(1)
\end{equation*} 
}
\end{lemma}
\begin{proof}
The proof follows as on pages 1636-1638 of \citeasnoun{Robinson95LW} using the bound from Lemma \ref{Periodogram_Mildly}. The other bounds in (3.17) of \citeasnoun{Robinson95LW} can be computed adapting arguments in Theorem 2 of \citeasnoun{Robinson95log-p} as we did for Lemma \ref{Periodogram_Mildly}.
\end{proof}

\textbf{Proof of Theorem}~\ref{DMforMildly}  - \textit{Case P} 
\begin{proof} We rewrite $\sqrt{T}\frac{\overline{d}-\mu}
   {\widehat{\sigma}_P}$ as
   \begin{equation*}
       \sqrt{T}\frac{\overline{d}-\mu}
   {\widehat{\sigma}_P}
  =\frac{\sqrt{({2 \pi f(0)})^{-1}}
   \sqrt{T}(\overline{d}-\mu)}
   {\sqrt{({2 \pi f(0)})^{-1}}
   \sqrt{ \frac{2 \pi}{m} \sum_{j=1}^m{I(\lambda_j)} }}
   \end{equation*} 
   and recall ${2 \pi f(0)}=\upsilon^{-2} \omega^2 =
   (-c)^{-2} T^{2 \alpha} \omega^2$. \\ From Theorem 2.1 of \citeasnoun{GiraitisPhillips2012Local}, $\sqrt{({2 \pi f(0)})^{-1}}
   \sqrt{T}(\overline{d}-\mu)
   \rightarrow_d N(0,1)$ .
   \\
As for the denominator, we discuss the cases 
$m T^{\alpha-1} \rightarrow 0$ and $m T^{\alpha-1} \rightarrow \infty$ separately.\\
We begin with $m T^{\alpha-1} \rightarrow 0$.
\\ 
We rewrite the argument of the square root of the denominator as 
\begin{align*}
 &    f(0)^{-1} \frac{1}{m} \sum_{j=1}^m{I(\lambda_j)}
    =     f(0)^{-1} \frac{1}{m} 
    \sum_{j=1}^m{ f(\lambda_j)^{-1} I(\lambda_j)f(\lambda_j)} 
    =
    f(0)^{-1} \frac{1}{m} 
    \sum_{j=1}^m{ (f(\lambda_j)^{-1} I(\lambda_j)-1+1)f(\lambda_j)} \\
&    =f(0)^{-1} \frac{1}{m} 
    \sum_{j=1}^m{ (f(\lambda_j)^{-1} I(\lambda_j)-1)} f(\lambda_j)
    +f(0)^{-1} \frac{1}{m} 
    \sum_{j=1}^m{f(\lambda_j)}.
\end{align*}
The first term is bounded as
\begin{align}
    & \left|f(0)^{-1} \frac{1}{m} 
    \sum_{j=1}^m{ (f(\lambda_j)^{-1} I(\lambda_j)-1)}f(\lambda_j)\right| \nonumber   \\
&    \leq f(0)^{-1} \frac{1}{m} 
    \sum_{j=1}^{m-1}
{|f(\lambda_j)-f(\lambda_{j+1})| \left|\sum_{k=1}^{j} (f(\lambda_k)^{-1} I(\lambda_k)-1)\right|}+
 f(0)^{-1} f(\lambda_m) \frac{1}{m}\left|\sum_{j=1}^{m} {(f(\lambda_j)^{-1} I(\lambda_j)-1)}\right| \label{Thm3Bound1} 
\end{align}
The first term of \eqref{Thm3Bound1} has the same order as 
\begin{equation}
\label{Thm3Bound2}
f(0)^{-1} \frac{1}{m} 
    \sum_{j=1}^{m-1}
{|f^\star(\lambda_j)'| \frac{1}{T} \left|\sum_{k=1}^{j} (f(\lambda_k)^{-1} I(\lambda_k)-1)\right|} 
+
f(0)^{-1} \frac{1}{m} 
    \sum_{j=1}^{m-1}
f^\star(\lambda_j){ \frac{1}{T} \left|\sum_{k=1}^{j} (f(\lambda_k)^{-1} I(\lambda_k)-1)\right|} 
\end{equation}
 the first element in \eqref{Thm3Bound2} 
 has order 
\begin{align*}
&    f(0)^{-1} \frac{1}{m} 
    \sum_{j=1}^{m-1}
{|(f^\star(\lambda_j))'| \frac{1}{T} j \ j^{-1} \left|\sum_{k=1}^{j} (f(\lambda_k)^{-1} I(\lambda_k)-1)\right|} \\
 &   \leq C
    f^\star(0)^{-1} \frac{1}{m} 
    \sum_{j=1}^{m-1}
{f^\star(0) \lambda_j^{-2} \lambda_j \frac{1}{T} j \ j^{-1} \left|\sum_{k=1}^{j} (f(\lambda_k)^{-1} I(\lambda_k)-1)\right|} 
=o_p\left(\frac{1}{m} 
    \sum_{j=1}^{m-1}
{\lambda_j^{-1} \frac{1}{T} j}\right)
=o_p(1)
\end{align*}
where we used the bound $|f'^\star(\lambda_j)'| 
\leq C f^\star(\lambda_j)^2 \lambda_j 
\leq C f^\star(0) f^\star(\lambda_j) \lambda_j
\leq C f^\star(0) \lambda_j^{-2} \lambda_j $ and Lemma \ref{Daniell_Mildly}; the second element in \eqref{Thm3Bound2} 
 has order 
\begin{align*}
&    f(0)^{-1} \frac{1}{m} 
    \sum_{j=1}^{m-1}
f^\star(\lambda_j){ \ j \ j^{-1} \frac{1}{T} \left|\sum_{k=1}^{j} (f(\lambda_k)^{-1} I(\lambda_k)-1)\right|}  \\
 &   \leq C
     \frac{1}{m} 
    \sum_{j=1}^{m-1}
 { \frac{1}{T} j \ j^{-1} \left|\sum_{k=1}^{j} (f(\lambda_k)^{-1} I(\lambda_k)-1)\right|} 
=o_p\left(\frac{1}{m} 
    \sum_{j=1}^{m-1}
{ \frac{1}{T} j}\right)
=o_p\left(\frac{m}{T}\right)=o_p(1).
\end{align*}
so the bound in \eqref{Thm3Bound2} is $o_p(1)$. The second term in \eqref{Thm3Bound1} can be discussed in the same way, thus establishing 
that the bound in \eqref{Thm3Bound1}  is  $o_p(1)$. \\
To complete the discussion of the denominator when $m \ T^{\alpha-1} \rightarrow 0$, we need to consider the term $f(0)^{-1} \frac{1}{m} 
\sum_{j=1}^m f(\lambda_j)$. 
This is 
\begin{equation*}
    f(0)^{-1} \frac{1}{m} 
\sum_{j=1}^m f(\lambda_j)
=
    f(0)^{-1} \frac{1}{m} 
\sum_{j=1}^m (f(\lambda_j)-f(0))+1
\end{equation*}
and notice that
\begin{equation*} 
f(0)^{-1} \frac{1}{m} \sum_{j=1}^m |f(\lambda_j)-f(0)| 
\leq C
f(0)^{-1} \frac{1}{m}
\sum_{j=1}^m |\lambda_j^2 f^\star(\lambda_j) \upsilon^{-2}|
+
C
f(0)^{-1} \frac{1}{m}
\sum_{j=1}^m |\lambda_j^2  \upsilon^{-2}|
=O(T^{2\alpha-2 }m^2)=o(1),
\end{equation*} 
Combining these results, 
\begin{equation*}
f(0)^{-1} \frac{1}{m} \sum_{j=1}^m{I(\lambda_j)}  \rightarrow_p 1,
\end{equation*}
completing the discussion for the case 
$T^{\alpha-1} \ m \rightarrow 0$. \\
\\
For the $T^{\alpha-1} \ m \rightarrow \infty$, again we only need to consider the denominator. The argument of the square root in this case, is 
\begin{equation*}
\{(-c)^2 T^{-2 \alpha} \omega^2\}^{-1}    
   \frac{2 \pi}{m}
   \sum_{j=1}^m{I(\lambda_j)} 
\end{equation*}
and notice that 
\begin{equation*}
    \frac{2 \pi}{T \ T^{\alpha} }\sum_{j=1}^m{I(\lambda_j)}=
     \frac{1}{2}
     \frac{1}{T^{1+\alpha}} \sum_{t=1}^T {(d_t-\overline{d})^2}
     -\frac{2 \pi}{T \ T^{\alpha} }\sum_{j=m+1}^{T/2}{I(\lambda_j)}
\end{equation*}
From Lemma \ref{BoundPeriodogramMildly},
\begin{equation}
    \frac{2 \pi}{T \ T^{\alpha} }\sum_{j=m+1}^{T/2}{I(\lambda_j)}=
    O_p \left( \frac{1}{T \ T^{\alpha} }\sum_{j=m+1}^{T/2}{\lambda_j^{-2}} \right)=O_p \left( T^{1-\alpha} m^{-1} \right)=o_p(1)
\end{equation}
The application of Lemma \ref{LimitVariance} for $\frac{1}{T^{1+\alpha}} \sum_{t=1}^T {(d_t-\overline{d})^2}$ completes the proof.
\end{proof}

\textbf{Proof of Theorem}~\ref{DMforMildly}  - \textit{Case A} 
\begin{proof}
The proof uses the same decomposition \eqref{gamma_l} and the limit in Lemma \ref{LimitVariance}; then, for all $l \leq M$
\begin{equation*}
\frac{1}{T^{\alpha}} \widehat{\gamma}_l=\frac{1}{T^{1+\alpha}} \sum_{t=1}^T (y_t-\overline{y})+O_p(M/T^\alpha)+O_p(M^{1/2}/T^{\alpha/2})
\end{equation*} 
and therefore
\begin{equation*}
\frac{1}{M T^{\alpha}}\widehat{\sigma}_A^2 \rightarrow_p \frac{\omega^2}{-c} 
\end{equation*}
The statement of the theorem follows from the CLT for $\overline{d}$ as in the proof for Case P and the continuous mapping theorem.
\end{proof}

\section{Additional plots}\label{app_plots}
In this appendix, we report additional plots for the empirical application. In particular, Figure~\ref{fig:err} plots the realised forecast errors from  AR(1) forecasts along with the $2\%$  and the rolling average benchmarks for forecasting horizons 2, 4, 6 and 8 quarters. As in the main part of the paper, forecast errors are defined as the realised value minus the prediction. The figure indicates that, for short forecasting horizons, AR(1) forecasts are more accurate and have less persistent forecast errors than the benchmarks. The realised losses of the three forecasts using a quadratic loss function reported in Figure~\ref{fig:loss} confirm that, for short forecasting horizons, AR(1) forecasts are more precise than the benchmarks and have losses that are not too correlated. However, when we look at the loss differentials reported in Figure~\ref{fig:lossd}, we see that they inherit the dependence properties of the benchmarks and display relevant autocorrelations,
even at short forecasting horizons.

\begin{figure}[H]
\caption{Realised forecast errors \bigskip} 
\includegraphics[trim={1.5cm 8cm 1cm 8cm},clip,scale=0.9]{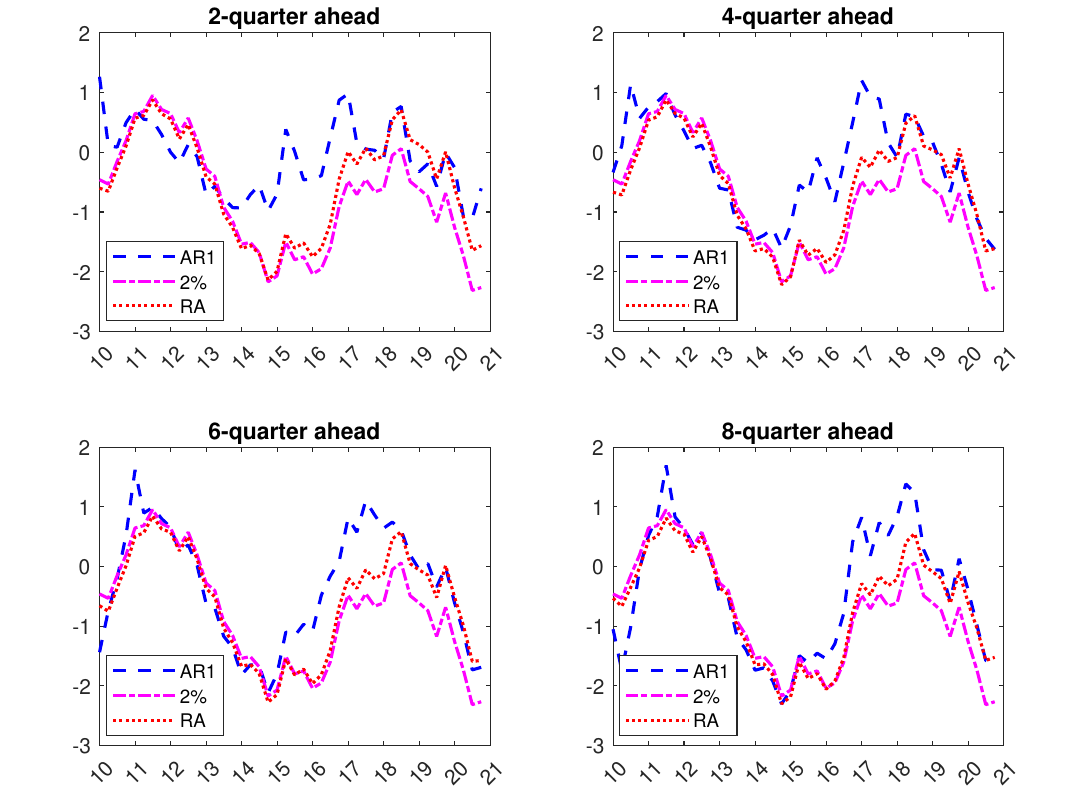}
\label{fig:err}
Note: realised forecast errors for AR(1) forecasts(AR1), along with the $2\%$  ($2\%$) and the rolling average (RA) benchmarks for forecasting horizons 2, 4, 6 and 8 quarters. Forecast errors are defined as the realised value minus the prediction.
\end{figure}

\begin{figure}[H]
\caption{Realised forecast losses \bigskip} 
\includegraphics[trim={1.5cm 8cm 1cm 8cm},clip,scale=0.9]{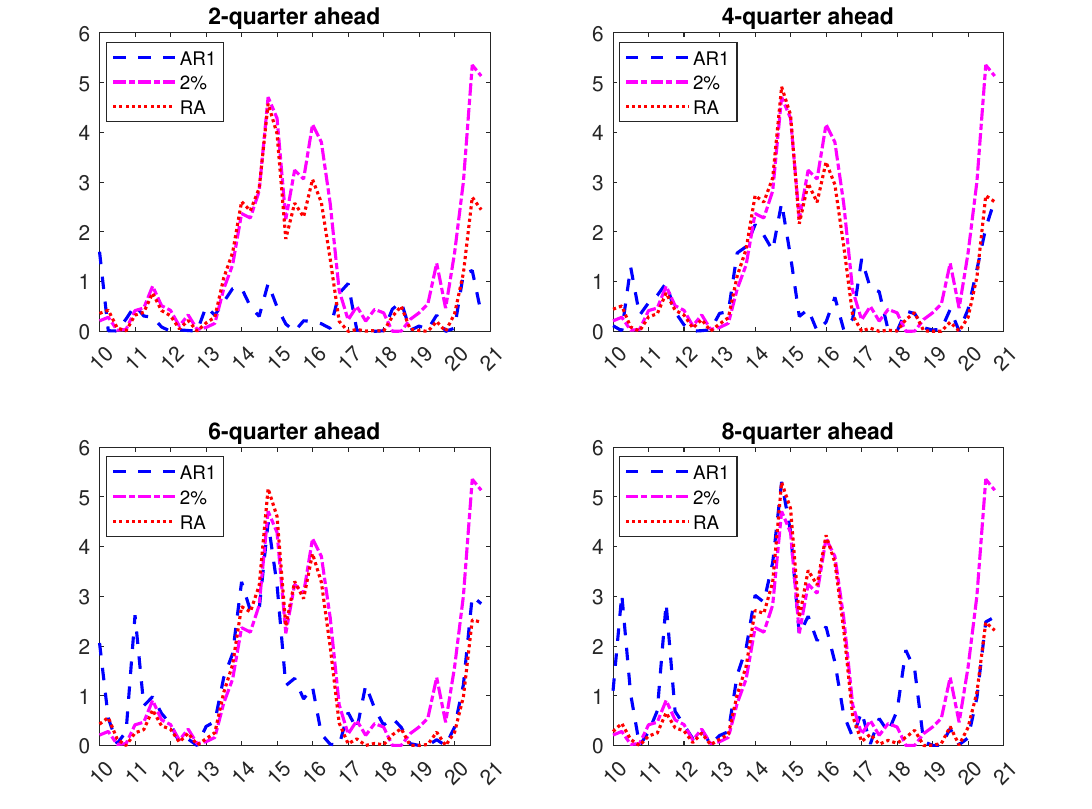}
\label{fig:loss}
Note: realised forecast losses for AR(1) forecasts (AR1), along with the $2\%$  ($2\%$) and the rolling average (RA) benchmarks for forecasting horizons 2, 4, 6 and 8 quarters. The loss function is quadratic.
\end{figure}

\begin{figure}[H]
\caption{Realised loss differentials \bigskip} 
\includegraphics[trim={1.5cm 8cm 1cm 8cm},clip,scale=0.9]{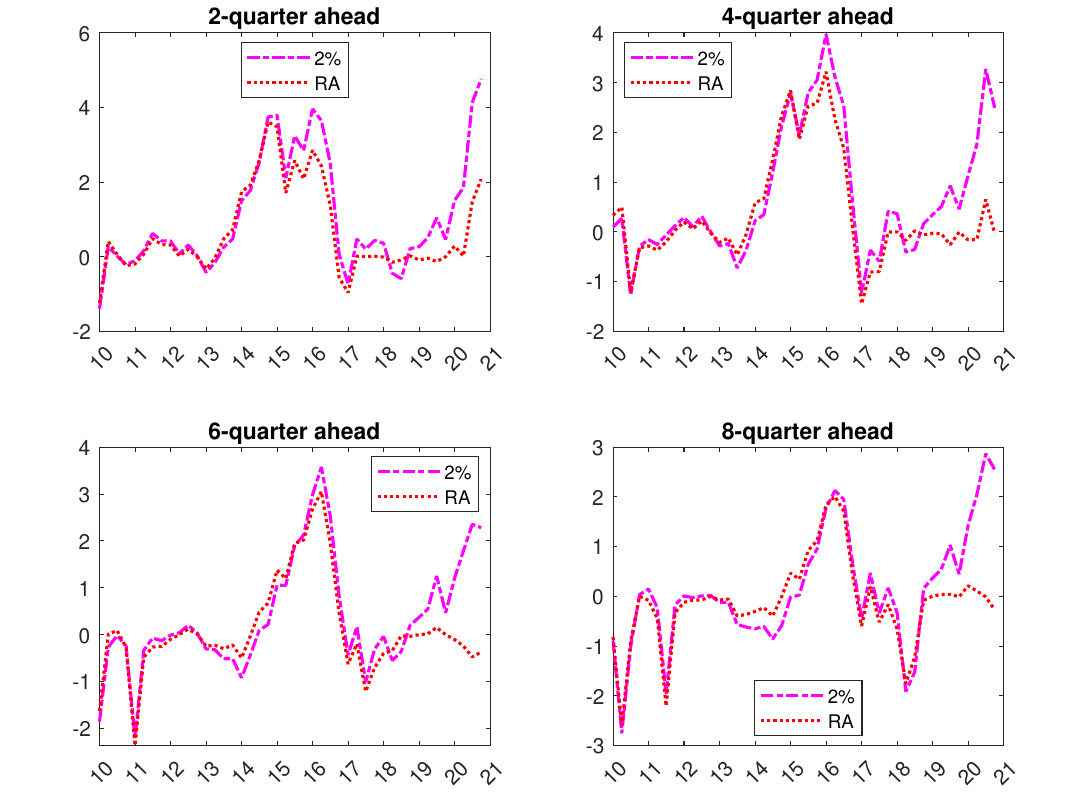}
\label{fig:lossd}
Note: realised loss differentials for AR(1) forecasts (AR1) with respect to the $2\%$  ($2\%$) and the rolling average (RA) benchmarks for forecasting horizons 2, 4, 6 and 8 quarters. The loss function is quadratic, and the loss differential is computed
as the loss of the benchmark minus the loss of the AR(1) forecast.
\end{figure}
\end{document}